\documentclass[letterpaper,11pt]{article}
\usepackage{amsmath, amsthm, amssymb}
\usepackage[usenames, dvipsnames]{color}
\usepackage[normalem]{ulem} 
\usepackage{fullpage}
\usepackage[style=alphabetic,backend=bibtex,maxnames=100,giveninits=true]{biblatex}
\usepackage[noend]{algorithmic}
\usepackage{tikz, pgfplots}
\usepackage{enumerate}
\usepackage[hidelinks,colorlinks,citecolor=blue]{hyperref}
\usepackage{xspace,color}
\usepackage{graphicx}
\usepackage{caption}
\usepackage{subcaption}
\usepackage{enumitem,linegoal}
\usepackage[linesnumbered,ruled,vlined]{algorithm2e}
\usepackage{comment}	
\usepackage{ctable}

\usepackage{mathtools}
\usepackage{mathrsfs}
\usepackage[flushleft]{threeparttable}
\usepackage{verbatim}
\usepackage{setspace}
\usepackage{bbm}
\usepackage{thm-restate}
\usepackage{footnote}
\usepackage{cleveref}
\usepackage{nicefrac}

\usepackage{amsmath, amsthm}
\usepackage{booktabs} 
\usepackage{tikz}
\usepackage[bibliography=common,appendix=inline]{apxproof}
\usepackage{etoolbox}
\usepackage[noend]{algorithmic}

\crefname{LP}{LP}{LPs}
\crefname{ineq}{inequality}{inequalities}
\creflabelformat{ineq}{#2{\upshape(#1)}#3} 
\makesavenoteenv{tabular}
\makesavenoteenv{table}
\newcommand\numberthis{\addtocounter{equation}{1}\tag{\theequation}}
\usepackage[margin=1in]{geometry}

\definecolor{crimsonglory}{rgb}{0,0,0}

 \newtheorem{theorem}{Theorem}[section]
 
 \newtheorem{lemma}[theorem]{Lemma}

 \newtheorem{definition}[theorem]{Definition}

\makeatletter
\def\GrabProofArgument[#1]{ #1: \egroup\ignorespaces}
\def\proof{\noindent\textbf\bgroup Proof%
	\@ifnextchar[{\GrabProofArgument}{. \egroup\ignorespaces}}

\makeatother

\usetikzlibrary{arrows,shapes,snakes,automata,backgrounds,petri,calc}
\usepackage[latin1]{inputenc}

\newcounter{proccnt}

\newcommand{\konote}[1]{}

\usepackage[normalem]{ulem} 

\newtheoremrep{theorem}{Theorem}[section]
\newtheoremrep{lemma}[theorem]{Lemma}

\SetAlFnt{\small}
\SetAlCapFnt{\small}
\SetAlCapNameFnt{\small}
\SetAlCapHSkip{0pt}
\IncMargin{-\parindent}

\crefname{LP}{LP}{LPs}
\crefname{ineq}{inequality}{inequalities}
\creflabelformat{ineq}{#2{\upshape(#1)}#3}

\title{Beating the Logarithmic Barrier for the Subadditive Maximin Share Problem}

\author{
    Masoud Seddighin\thanks{Tehran Institute for Advanced Studies} 
	\and Saeed Seddighin
}

\SetCommentSty{mycommfont}

\SetKwInput{KwInput}{Input}                
\SetKwInput{KwOutput}{Output}  

\newcommand{\MMS}{\textsf{MMS}}
\newcommand{\valu}{V}
\newcommand{\items}{M}
\newcommand{\agents}{N}
\newcommand{\alloc}{A}
\newcommand{\agent}{a}
\newcommand{\matching}{\mathscr{M}}
\newcommand{\easily}{\mathcal E}
\newcommand{\noteasily}{\mathcal H}
\newcommand{\block}{B}

\newcommand{\subsett}{X}
\newcommand{\ite}{b}

\newcommand{\tr}[1]{{\hat{#1}}}

\begin{document}
	\newcommand{\ignore}[1]{}
\renewcommand{\theenumi}{(\roman{enumi}).}
\renewcommand{\labelenumi}{\theenumi}
\sloppy
\date{}

\maketitle

\thispagestyle{empty}

\begin{abstract}
We study the problem of fair allocation of indivisible goods for subadditive agents. 
While constant-\textsf{MMS} bounds have been given for additive and fractionally subadditive agents, the best existential bound for the case of subadditive agents is $1/O(\log n \log \log n)$. In this work, we improve this bound to a $1/O((\log \log n)^2)$-\textsf{MMS} guarantee. To this end, we introduce new matching techniques and rounding methods for subadditive valuations that we believe are of independent interest and will find their applications in future work.
\end{abstract}
\section{Introduction}\label{introduction}

Fair division is a classic problem in mathematics and economics, with applications that extend to fields such as political science, social science, and computer science \cite{Dubins:first,Steinhaus:first,brams1996fair}. The problem was first formally introduced by Hugo Steinhaus under the title of "cake cutting," which aims to divide a divisible resource -- known as the "cake" -- among a group of agents with differing preferences. Over the past eighty years, the cake-cutting problem has been a major focus of research \cite{even1984note,stromquist1980cut,brams1996fair}. More recently, a discrete version of this problem, which involves the allocation of indivisible items, has gained attention, particularly in the field of computer science \cite{Procaccia:first,ghodsi2018fair,caragiannis2016unreasonable,amanatidis2015approximation,kaviani2024almost,Budish:first,Saberi:first,robertson1998cake,seddighin2019externalities}.

A key challenge of fair division is defining fairness in a way that is both meaningful and achievable. A fairness notion must strike a balance between feasibility and what participants perceive as fair. In cake-cutting, notions such as proportionality and envy-freeness capture this balance well --- they align with intuitive fairness perceptions and have been widely accepted and extensively studied. Proportionality ensures that each participant receives at least their $1/n$ of the total value in an $n$-person fair division, while envy-freeness guarantees that no one prefers another participant's allocation over their own. Many results show that these notions can be guaranteed, often alongside other desirable properties such as connectivity of shares and Pareto optimality \cite{Dubins:first,stromquist1980cut,even1984note}.

Moving beyond cake-cutting, these fairness notions become impractical when dealing with indivisible items. Consider a simple example: if there is only one indivisible item and two agents, neither envy-freeness nor proportionality -- nor even any approximation of these notions -- can be guaranteed. This shortcoming has led to the development of more flexible fairness criteria. 

Over the past decade, several such notions have emerged, including envy-freeness up to one item (\textsf{EF1}), envy-freeness up to any item (\textsf{EFX}), envy-freeness up to a random item (\textsf{EFR}), proportionality up to one item (\textsf{Prop1}), maximin share (\textsf{MMS}), pairwise maximin share (\textsf{PMMS}) \cite{Budish:first,caragiannis2016unreasonable,conitzer2017fair,farhadi2017fair}. Among them, \textsf{EFX} is widely regarded as the most prominent alternative to envy-freeness, while \textsf{MMS} is the most well-studied relaxation of proportionality.

In this paper, we focus on the maximin-share  criterion. This notion was first introduced by \textcite{Budish:first} as a relaxation of proportionality for indivisible items. Roughly speaking, the \textsf{MMS} value of an agent represents the best guarantee they can secure if they were to divide the items into \( n \) bundles and receive the least valued bundle under their own partitioning.

Formally, for a set of items $\items$ and an agent $\agent_i$ among $n$ agents, the maximin share value of agent $\agent_i$ for items in $\items$, denoted by $\MMS_i(\items)$, is defined as
$
\MMS_i(\items) = \max_{\langle \pi_1, \pi_2, \ldots, \pi_n \rangle \in \Pi} \min_{1 \leq j \leq n} \valu_i(\pi_j),
$
where $\Pi$ represents the set of all possible partitionings of $\items$ into $n$ bundles, and $\valu_i(\pi_j)$ is the value that agent $\agent_i$ assigns to bundle $\pi_j$. An allocation is called \textsf{MMS} (or $\beta$-\textsf{MMS}) if every agent $\agent_i$ receives a bundle worth at least $\MMS_i(\items)$ (or $\beta \cdot \MMS_i(\items)$) to her.

Maximin share (\(\MMS\)) was first introduced to computer science by \textcite{Procaccia:first}. In their work, they present an elegant counterexample showing that \(\MMS\) allocations do not always exist. On the positive side, they prove that for additive valuations, a \(\nicefrac{2}{3}\)-\(\MMS\) allocation can always be guaranteed \cite{Procaccia:first}.  
This bound is later improved to \(\nicefrac{3}{4}\)-\(\MMS\) by \textcite{ghodsi2018fair}. Subsequent refinements, including \(\left(\nicefrac{3}{4} + \nicefrac{1}{12n}\right)\)-\(\MMS\) by \textcite{garg2020improved} and \(\left(\nicefrac{3}{4} + \min\left(\nicefrac{1}{36}, \nicefrac{3}{16n-4}\right)\right)\)-\(\MMS\) by \textcite{akrami2023simplification}, provide slight improvements but converge to \(\nicefrac{3}{4}\)-\(\MMS\) for large \( n \). Recently, \textcite{akrami2024breaking} break this barrier by proving a guaranteed factor of \(\left(\nicefrac{3}{4} + \nicefrac{3}{3836}\right)\)-\(\MMS\), marking the first improvement beyond \(\nicefrac{3}{4}\)-\(\MMS\) independent of \( n \).

In addition to the additive setting, it has be shown that under fractionally subadditive setting also a constant approximation guarantee is possible~\cite{maseed123,ghodsi2018fair,akrami2024randomized}. However, this is in contrast to the subadditive setting for which the best bound prior to our work is $\nicefrac{1}{O(\log n \log \log n)}$-\textsf{MMS} guarantee~\cite{maseed123}. In this work we improve this bound to $\nicefrac{1}{O((\log \log n)^2)}$-\textsf{MMS} guarantee. Similar to previous work, we first show the existence of a multiallocation (an allocation in which an item may be given to multiple agents) with provable guarantees and then show that our multiallocation can be turned into a desired allocation. Conventional methods would lose an $O(\log n)$ in each of these steps however, we manage to introduce new techniques that only lose $O(\log \log n)$ in each step.
In particular, parts of our analysis extend and generalize the rounding method of~\textcite{feige2009maximizing} which we believe is of independent interest and will find its applications in future work.


\section{Preliminaries}\label{prelim}

We define the set of agents as \(\agents = \{\agent_1, \agent_2, \ldots, \agent_n\}\) and the set of items as \(\items = \{\ite_1, \ite_2, \ldots, \ite_m\}\). Each agent \(\agent_i\) assigns a valuation \(\valu_i(S)\) to any set \(S \subseteq \items\). We assume that valuations are \textbf{non-negative}, meaning \(\valu_i(S) \geq 0\) for all agents \(\agent_i\) and sets \(S\), and \textbf{monotone}, meaning that for any two sets \(S_1, S_2 \subseteq \items\) and for all 
$\agent_i \in \agents$, $\valu_i(S_1 \cup S_2) \geq \max\{\valu_i(S_1), \valu_i(S_2)\}$ holds. also, we assume that the valuation functions are \textbf{subadditive}, so for every agent \(\agent_i\) and any two sets of items \(S_1, S_2\), we have  
$
V(S_1) + V(S_2) \geq V(S_1 \cup S_2).
$ 
Let $\Pi_r$ be the set of all partitionings of $\items$ into $r$ disjoint subsets. For a  set function $f(\cdot)$, we define $\MMS_f^r(\items)$ as 
$ \MMS_f^r(\items) = \max_{\langle \pi_i,\pi_2,\ldots,\pi_r \rangle \in {\Pi_r}}  \min_{1 \leq j \leq r} f(\pi_j).$

A \textbf{multiallocation} of items to agents is an $n$-tuple $\mathcal{A} = \langle A_1, A_2, \ldots, A_n \rangle$, where $A_i \subseteq \items$ is the bundle of items allocated to agent $\agent_i$. Note that the bundles $A_i$ do not need to be disjoint. We say that $\mathcal{A}$ is an \textbf{$\alpha$-multiallocation} if each item appears in at most $\alpha$ different bundles. When $\alpha = 1$, we refer to a $1$-multiallocation simply as an \textbf{allocation}. In an allocation, the bundles allocated to any two agents $\agent_i$ and $\agent_j$ are disjoint.

An allocation (or multiallocation) $\mathcal{A}$ is called \textbf{$\alpha$-$\MMS$} if every agent $\agent_i$ receives a subset of items whose value to her is at least $\alpha$ times their maximin-share ($\MMS_i$). Formally, $\mathcal{A}$ is $\alpha$-$\MMS$ if and only if
$
\valu_i(A_i) \geq \alpha \cdot \MMS_i
$
holds for every agent $\agent_i \in \agents$.

Finally, since maximin-share is a scale-free criterion, we assume for simplicity throughout this paper that $\MMS_{\valu_i}^n(\items)=1$ for every agent $\agent_i \in \agents$. 

\section{Our Results}
We prove that under subadditive valuation functions, there always exists an allocation that is $\nicefrac{1}{O((\log \log n)^2)}$-\MMS. As mentioned earlier, our proof is consisted of two parts. We first show a reduction that shows the existence of a multiallocation with bounded \MMS\ guarantees would result in an allocation with similar guarantees. We then proceed by presenting algorithms that are guaranteed to produce desirable multiallocations with non-zero probability. We outline these steps seperately in \Cref{sec:o1,sec:o2}. 
\subsection{Reduction from Allocations to Multiallocations}\label{sec:o1}

Let us begin by stating a concentration bound for which we bring a proof in \Cref{sec:first}.

\vspace{0.2cm}
{\noindent \textbf{Lemma} \ref{lemma:1} [restated informally]. \textit{Let for an $\hat{n} \geq 1$ and a ground set of elements $\hat{M}$, $f:2^{\hat{M}} \rightarrow \mathbb{R}_+$ be a monotone subadditive function with non-negative values such that $f(S) \leq f(\hat{M})/2$ holds for any subset $S \subseteq \hat{M}$ such that $|S| \leq O(\nicefrac{\log \hat{n}}{p})$. Let $R$ be a random subset of $\hat{M}$ such that each element of $\hat{M}$ appears in $R$ independently with probability $0 \leq p \leq 1$. Then for some constant $c = \Omega(1)$ we have:
		$\mathsf{Pr} \left[ f(R) \geq c f(\hat{M})p \right] > 1-1/\hat{n}.$\\}}

Lemma~\ref{lemma:second} presents a convenient tool to turn a multiallocation into a valid allocation by losing an additional $O(\log \log n)$ multiplicative factor in the approximation. 

\vspace{0.2cm}
{\noindent \textbf{Lemma} \ref{lemma:second} [restated informally]. \textit{Let $\hat{\valu}_1, \hat{\valu}_2, \ldots, \hat{\valu}_{\hat{n}}$ be $\hat{n}$ monotone subadditive functions with non-negative valuations defined on a ground set of elements $\hat{M} = \{\hat{\ite}_1, \hat{\ite}_2, \hat{\ite}_3, \ldots, \hat{\ite}_{|\hat{M}|}\}$. Let $A_1, A_2, \ldots, A_{\hat{n}} \subseteq \hat{M}$ be $\hat{n}$ subsets of $\hat{M}$ such that no element of $\hat{M}$ appears in more than $\alpha$ of these subsets. If $\hat{\valu}_{i}({\{\hat{\ite}_x\}}) \leq \beta$ for every $\hat{\ite}_x \in \hat{M}$ and $1 \leq i \leq \hat{n}$ then there exist $\hat{n}$ disjoint subsets $A'_1, A'_2, \ldots, A'_{\hat{n}} \subseteq \hat{M}$ such that for some $c_1 = \Omega(1)$ and  $c_2 = O(1)$ we have:
		\begin{equation}\label{eq:lem2} \hat{\valu}_i(A'_i) \geq \frac{c_1 \hat{\valu}_i(A_i)}{\alpha (\log \log \hat{n} + \log \alpha)} - c_2 \beta.\end{equation}\\}}
To better understand the connection between Lemma~\ref{lemma:second} and the allocation problem, consider $n$ agents $\agent_1, \agent_2, \ldots, \agent_n$ whose valuations functions for items of $M$ are $\valu_1, \valu_2, \ldots, \valu_n$ and their allocated items in a multiallocation are $\mathcal{A}= A_1, \ldots, A_n$. \Cref{lemma:second} proves that it is possible to obtain an allocation $\mathcal{A}' = A'_1, A'_2, \ldots, A'_n$ of items to the agents such that the utility of each agent is lower bounded by a fraction of her utility in the multiallocation. In what follows, we outline the ideas of \Cref{lemma:second}.

Remark that Lemma~\ref{lemma:second} goes beyond our allocation context and can be used in a general form. Therefore, we use $\hat{n}$, $\hat{\valu}$, and  $\hat{M}$ to distinguish between the inputs of Lemma~\ref{lemma:second} and the variables of the \MMS\ problem. However, we explain the ideas in terms of the allocation problem to simplify the understanding. That is, we assume $\hat{M} = \{\hat{\ite}_1, \hat{\ite}_2, \hat{\ite}_3, \ldots, \hat{\ite}_{|\hat{M}|}\}$ are items and there are $\hat{n}$ agents $\hat{\agent}_1,\hat{\agent}_2,\ldots,\hat{\agent}_{\hat{n}}$ such that each $\hat{\valu_{i}}$ represents the valuation of agent $\hat{\agent}_i$ for the items of $\hat{M}$. At a high-level, we divide the agents into two categories: 
\begin{itemize}
	\item Category (i) \textbf{easily-satisfiable agents}: An agent $\hat{\agent}_i$ is \textit{easily-satisfiable} if and only if there exists a subset $X_i \subseteq A_i$ such that $\hat{\valu}_i(X_i) \geq \hat{\valu}_i(A_i)/2$ and $|X_i| \leq 80 \alpha (\log \hat{n}+1)$. 
	\item Category (ii) \textbf{not-easily-satisfiable agents}: An agent who is not-easily-satisfiable.
\end{itemize}

To construct the final allocation we first make an intermediary multiallocation 
in which every item is allocated to at most two agents. More precisely, in this multiallocation, in each category of agents, the allocated items are disjoint. However, an item may be allocated to both an easily-satisfiable agent and a not-easily-satisfiable agent. For an easily-satisfiable agent $\hat{\agent}_i$, we show that we can approximate her valuation function for $X_i$ by a linear function $\widetilde{\valu}_i$ that loses only a factor of $O(\log \alpha + \log \log \hat{n})$ in comparison to $\hat{\valu}_i(X_i)$. In other words, we have $\widetilde{\valu}_i(X_i) \geq \hat{\valu}_i(X_i) / (\log \alpha + \log \log \hat{n})/c$ for some constant $c$ and for each subset $Y \subseteq X_i$ we have $\widetilde{\valu}_i(Y) \leq \hat{\valu}_i(Y)$. We call $\widetilde{\valu}_i$ the auxiliary valuation function for agent $\hat{\agent}_i$. We then leverage the linearity of the auxiliary valuation functions and via a matching-based algorithm construct an allocation with the desired properties. 

\begin{figure}[t]
	
\tikzset{every picture/.style={line width=0.75pt}} 
\centering
\scalebox{1}{
\begin{tikzpicture}[x=0.75pt,y=0.75pt,yscale=-0.65,xscale=0.65]

\draw  [fill={rgb, 255:red, 126; green, 211; blue, 33 }  ,fill opacity=0.4 ] (259,346) -- (303,346) -- (303,386) -- (259,386) -- cycle ;
\draw  [fill={rgb, 255:red, 189; green, 16; blue, 224 }  ,fill opacity=0.4 ] (319,346) -- (363,346) -- (363,386) -- (319,386) -- cycle ;
\draw  [fill={rgb, 255:red, 74; green, 144; blue, 226 }  ,fill opacity=0.4 ] (377,346) -- (421,346) -- (421,386) -- (377,386) -- cycle ;
\draw  [fill={rgb, 255:red, 248; green, 231; blue, 28 }  ,fill opacity=0.4 ] (457,346) -- (501,346) -- (501,386) -- (457,386) -- cycle ;
\draw  [fill={rgb, 255:red, 139; green, 87; blue, 42 }  ,fill opacity=0.4 ] (517,346) -- (561,346) -- (561,386) -- (517,386) -- cycle ;
\draw  [fill={rgb, 255:red, 208; green, 2; blue, 27 }  ,fill opacity=0.4 ] (575,346) -- (619,346) -- (619,386) -- (575,386) -- cycle ;
\draw   (17,332) -- (642,332) -- (642,397) -- (17,397) -- cycle ;
\draw  [fill={rgb, 255:red, 0; green, 0; blue, 0 }  ,fill opacity=0.4 ] (30,346) -- (74,346) -- (74,386) -- (30,386) -- cycle ;
\draw  [fill={rgb, 255:red, 155; green, 155; blue, 155 }  ,fill opacity=0.4 ] (90,346) -- (134,346) -- (134,386) -- (90,386) -- cycle ;
\draw  [fill={rgb, 255:red, 245; green, 166; blue, 35 }  ,fill opacity=0.4 ] (148,346) -- (192,346) -- (192,386) -- (148,386) -- cycle ;
\draw  [dash pattern={on 4.5pt off 4.5pt}] (24,311) -- (199,311) -- (199,393) -- (24,393) -- cycle ;
\draw  [dash pattern={on 4.5pt off 4.5pt}] (251,311) -- (426,311) -- (426,393) -- (251,393) -- cycle ;
\draw  [dash pattern={on 4.5pt off 4.5pt}] (444.71,312.5) -- (630.71,311.49) -- (631.15,392.5) -- (445.15,393.5) -- cycle ;
\draw  [fill={rgb, 255:red, 126; green, 211; blue, 33 }  ,fill opacity=0.4 ] (497,107) -- (541,107) -- (541,147) -- (497,147) -- cycle ;
\draw  [fill={rgb, 255:red, 189; green, 16; blue, 224 }  ,fill opacity=0.4 ] (381,108) -- (425,108) -- (425,148) -- (381,148) -- cycle ;
\draw  [fill={rgb, 255:red, 74; green, 144; blue, 226 }  ,fill opacity=0.4 ] (439,108) -- (483,108) -- (483,148) -- (439,148) -- cycle ;
\draw  [fill={rgb, 255:red, 248; green, 231; blue, 28 }  ,fill opacity=0.4 ] (323,108) -- (367,108) -- (367,148) -- (323,148) -- cycle ;
\draw  [fill={rgb, 255:red, 139; green, 87; blue, 42 }  ,fill opacity=0.4 ] (73,109) -- (117,109) -- (117,149) -- (73,149) -- cycle ;
\draw  [fill={rgb, 255:red, 208; green, 2; blue, 27 }  ,fill opacity=0.4 ] (554,106) -- (598,106) -- (598,146) -- (554,146) -- cycle ;
\draw  [fill={rgb, 255:red, 0; green, 0; blue, 0 }  ,fill opacity=0.4 ] (13,108) -- (57,108) -- (57,148) -- (13,148) -- cycle ;
\draw  [fill={rgb, 255:red, 155; green, 155; blue, 155 }  ,fill opacity=0.4 ] (200,108) -- (244,108) -- (244,148) -- (200,148) -- cycle ;
\draw  [fill={rgb, 255:red, 245; green, 166; blue, 35 }  ,fill opacity=0.4 ] (609,106) -- (653,106) -- (653,146) -- (609,146) -- cycle ;
\draw   (263,109) -- (307,109) -- (307,149) -- (263,149) -- cycle ;
\draw   (137,107) -- (181,107) -- (181,147) -- (137,147) -- cycle ;
\draw    (36,145) -- (111,321) ;
\draw [shift={(111,321)}, rotate = 66.92] [color={rgb, 255:red, 0; green, 0; blue, 0 }  ][fill={rgb, 255:red, 0; green, 0; blue, 0 }  ][line width=0.75]      (0, 0) circle [x radius= 3.35, y radius= 3.35]   ;
\draw [shift={(36,145)}, rotate = 66.92] [color={rgb, 255:red, 0; green, 0; blue, 0 }  ][fill={rgb, 255:red, 0; green, 0; blue, 0 }  ][line width=0.75]      (0, 0) circle [x radius= 3.35, y radius= 3.35]   ;
\draw    (221,139) -- (111,321) ;
\draw [shift={(111,321)}, rotate = 121.15] [color={rgb, 255:red, 0; green, 0; blue, 0 }  ][fill={rgb, 255:red, 0; green, 0; blue, 0 }  ][line width=0.75]      (0, 0) circle [x radius= 3.35, y radius= 3.35]   ;
\draw [shift={(221,139)}, rotate = 121.15] [color={rgb, 255:red, 0; green, 0; blue, 0 }  ][fill={rgb, 255:red, 0; green, 0; blue, 0 }  ][line width=0.75]      (0, 0) circle [x radius= 3.35, y radius= 3.35]   ;
\draw    (629,139) -- (111,321) ;
\draw [shift={(111,321)}, rotate = 160.64] [color={rgb, 255:red, 0; green, 0; blue, 0 }  ][fill={rgb, 255:red, 0; green, 0; blue, 0 }  ][line width=0.75]      (0, 0) circle [x radius= 3.35, y radius= 3.35]   ;
\draw [shift={(629,139)}, rotate = 160.64] [color={rgb, 255:red, 0; green, 0; blue, 0 }  ][fill={rgb, 255:red, 0; green, 0; blue, 0 }  ][line width=0.75]      (0, 0) circle [x radius= 3.35, y radius= 3.35]   ;
\draw    (404,143) -- (341,318) ;
\draw [shift={(341,318)}, rotate = 109.8] [color={rgb, 255:red, 0; green, 0; blue, 0 }  ][fill={rgb, 255:red, 0; green, 0; blue, 0 }  ][line width=0.75]      (0, 0) circle [x radius= 3.35, y radius= 3.35]   ;
\draw [shift={(404,143)}, rotate = 109.8] [color={rgb, 255:red, 0; green, 0; blue, 0 }  ][fill={rgb, 255:red, 0; green, 0; blue, 0 }  ][line width=0.75]      (0, 0) circle [x radius= 3.35, y radius= 3.35]   ;
\draw    (462,139) -- (341,318) ;
\draw [shift={(341,318)}, rotate = 124.06] [color={rgb, 255:red, 0; green, 0; blue, 0 }  ][fill={rgb, 255:red, 0; green, 0; blue, 0 }  ][line width=0.75]      (0, 0) circle [x radius= 3.35, y radius= 3.35]   ;
\draw [shift={(462,139)}, rotate = 124.06] [color={rgb, 255:red, 0; green, 0; blue, 0 }  ][fill={rgb, 255:red, 0; green, 0; blue, 0 }  ][line width=0.75]      (0, 0) circle [x radius= 3.35, y radius= 3.35]   ;
\draw    (520,139) -- (341,318) ;
\draw [shift={(341,318)}, rotate = 135] [color={rgb, 255:red, 0; green, 0; blue, 0 }  ][fill={rgb, 255:red, 0; green, 0; blue, 0 }  ][line width=0.75]      (0, 0) circle [x radius= 3.35, y radius= 3.35]   ;
\draw [shift={(520,139)}, rotate = 135] [color={rgb, 255:red, 0; green, 0; blue, 0 }  ][fill={rgb, 255:red, 0; green, 0; blue, 0 }  ][line width=0.75]      (0, 0) circle [x radius= 3.35, y radius= 3.35]   ;
\draw    (581,139) -- (535,321) ;
\draw [shift={(535,321)}, rotate = 104.18] [color={rgb, 255:red, 0; green, 0; blue, 0 }  ][fill={rgb, 255:red, 0; green, 0; blue, 0 }  ][line width=0.75]      (0, 0) circle [x radius= 3.35, y radius= 3.35]   ;
\draw [shift={(581,139)}, rotate = 104.18] [color={rgb, 255:red, 0; green, 0; blue, 0 }  ][fill={rgb, 255:red, 0; green, 0; blue, 0 }  ][line width=0.75]      (0, 0) circle [x radius= 3.35, y radius= 3.35]   ;
\draw    (96,139) -- (535,321) ;
\draw [shift={(535,321)}, rotate = 22.52] [color={rgb, 255:red, 0; green, 0; blue, 0 }  ][fill={rgb, 255:red, 0; green, 0; blue, 0 }  ][line width=0.75]      (0, 0) circle [x radius= 3.35, y radius= 3.35]   ;
\draw [shift={(96,139)}, rotate = 22.52] [color={rgb, 255:red, 0; green, 0; blue, 0 }  ][fill={rgb, 255:red, 0; green, 0; blue, 0 }  ][line width=0.75]      (0, 0) circle [x radius= 3.35, y radius= 3.35]   ;
\draw    (349,140) -- (535,321) ;
\draw [shift={(535,321)}, rotate = 44.22] [color={rgb, 255:red, 0; green, 0; blue, 0 }  ][fill={rgb, 255:red, 0; green, 0; blue, 0 }  ][line width=0.75]      (0, 0) circle [x radius= 3.35, y radius= 3.35]   ;
\draw [shift={(349,140)}, rotate = 44.22] [color={rgb, 255:red, 0; green, 0; blue, 0 }  ][fill={rgb, 255:red, 0; green, 0; blue, 0 }  ][line width=0.75]      (0, 0) circle [x radius= 3.35, y radius= 3.35]   ;

\draw (721,21) node    {$$};
\draw (701,71) node    {$$};
\draw (203,360.4) node [anchor=north west][inner sep=0.75pt]  [font=\Large]  {$\ldots$};

	\end{tikzpicture}}
	\caption{This figure represents a bipartite graph in which every node of one side is an item and each node of the other side in a set of consecutive elements in the allocated elements of an easily satisfiable agent.}\label{figure:matching}
\end{figure}
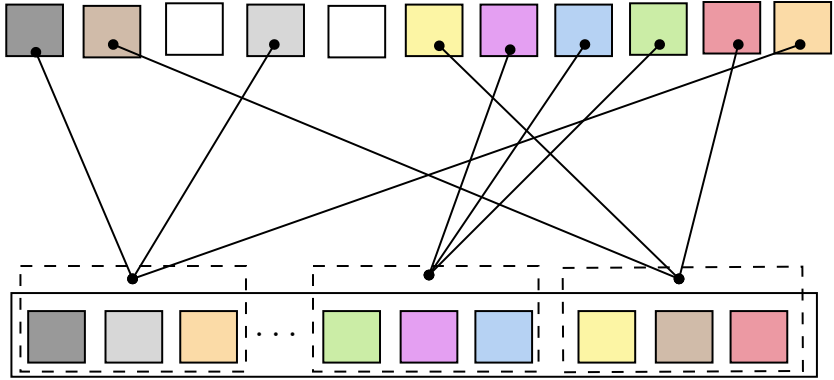

More precisely, we sort the allocated items of each easily-satisfiable agent based on their auxiliary linear function and make a group of every $\alpha$ consecutive elements in this order. (Recall that $\alpha$ is an upper bound on the number of bundles of $A_1, A_2, \ldots, A_{\hat{n}}$ that contain the same item). If the number of allocated items of an easily-satisfiable agent is not divisible by $\alpha$, we ignore the ones that have the least value in the auxiliary function to make their count divisible by $\alpha$. We then construct a bipartite graph as shown in \Cref{figure:matching}. In this graph, each node of the upper side represents an item and each node of the lower side is a group of consecutive $\alpha$ items in the allocated items of an  easily-satisfiable agent. We add an edge between two nodes if the corresponding item of the upper side node is in the group of items represented by the lower side node. We leverage the Hall theorem to show that the bipartite graph has a matching that covers all of the lower side nodes. The matching then shows how elements are given to the  easily-satisfiable agents. We show in the proof of Lemma~\ref{lemma:second} that such an allocation satisfies the desired guarantees of the easily-satisfiable agents.

For the not-easily-satisfiable agents, we use a random method to construct the allocation. Each item of the original multiallocation will be uniformly  and randomly allocated to one of the agents that received that item in the original multiallocation $A_1, A_2, \ldots, A_{\hat{n}}$. Using the concentration bound of Lemma~\ref{lemma:1}, we show that this allocation provides the desired guarantee with non-zero probability as well (therefore our desired allocation exists). By combining the two allocations made for the two categories, we obtain a 2-multiallocation 
for the agents. 

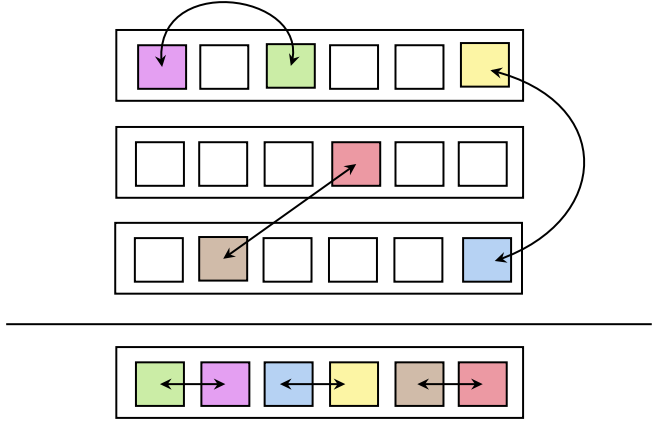
\begin{figure}[t]
	
	\centering
	\tikzset{every picture/.style={line width=0.75pt}} 
	\scalebox{1}{
	\begin{tikzpicture}[x=0.75pt,y=0.75pt,yscale=-0.55,xscale=0.55]
	
	\draw  [fill={rgb, 255:red, 126; green, 211; blue, 33 }  ,fill opacity=0.4 ] (159,346) -- (203,346) -- (203,386) -- (159,386) -- cycle ;
	\draw  [fill={rgb, 255:red, 189; green, 16; blue, 224 }  ,fill opacity=0.4 ] (219,346) -- (263,346) -- (263,386) -- (219,386) -- cycle ;
	\draw  [fill={rgb, 255:red, 74; green, 144; blue, 226 }  ,fill opacity=0.4 ] (277,346) -- (321,346) -- (321,386) -- (277,386) -- cycle ;
	\draw  [fill={rgb, 255:red, 248; green, 231; blue, 28 }  ,fill opacity=0.4 ] (337,346) -- (381,346) -- (381,386) -- (337,386) -- cycle ;
	\draw  [fill={rgb, 255:red, 139; green, 87; blue, 42 }  ,fill opacity=0.4 ] (397,346) -- (441,346) -- (441,386) -- (397,386) -- cycle ;
	\draw  [fill={rgb, 255:red, 208; green, 2; blue, 27 }  ,fill opacity=0.4 ] (455,346) -- (499,346) -- (499,386) -- (455,386) -- cycle ;
	\draw   (141,332) -- (514,332) -- (514,397) -- (141,397) -- cycle ;
	\draw    (420,366) -- (474,366) ;
	\draw [shift={(477,366)}, rotate = 180] [fill={rgb, 255:red, 0; green, 0; blue, 0 }  ][line width=0.08]  [draw opacity=0] (10.72,-5.15) -- (0,0) -- (10.72,5.15) -- (7.12,0) -- cycle    ;
	\draw [shift={(417,366)}, rotate = 0] [fill={rgb, 255:red, 0; green, 0; blue, 0 }  ][line width=0.08]  [draw opacity=0] (10.72,-5.15) -- (0,0) -- (10.72,5.15) -- (7.12,0) -- cycle    ;
	\draw    (294,366) -- (348,366) ;
	\draw [shift={(351,366)}, rotate = 180] [fill={rgb, 255:red, 0; green, 0; blue, 0 }  ][line width=0.08]  [draw opacity=0] (10.72,-5.15) -- (0,0) -- (10.72,5.15) -- (7.12,0) -- cycle    ;
	\draw [shift={(291,366)}, rotate = 0] [fill={rgb, 255:red, 0; green, 0; blue, 0 }  ][line width=0.08]  [draw opacity=0] (10.72,-5.15) -- (0,0) -- (10.72,5.15) -- (7.12,0) -- cycle    ;
	\draw    (184,366) -- (238,366) ;
	\draw [shift={(241,366)}, rotate = 180] [fill={rgb, 255:red, 0; green, 0; blue, 0 }  ][line width=0.08]  [draw opacity=0] (10.72,-5.15) -- (0,0) -- (10.72,5.15) -- (7.12,0) -- cycle    ;
	\draw [shift={(181,366)}, rotate = 0] [fill={rgb, 255:red, 0; green, 0; blue, 0 }  ][line width=0.08]  [draw opacity=0] (10.72,-5.15) -- (0,0) -- (10.72,5.15) -- (7.12,0) -- cycle    ;
	\draw   (158,232) -- (202,232) -- (202,272) -- (158,272) -- cycle ;
	\draw   (276,232) -- (320,232) -- (320,272) -- (276,272) -- cycle ;
	\draw   (336,232) -- (380,232) -- (380,272) -- (336,272) -- cycle ;
	\draw   (396,232) -- (440,232) -- (440,272) -- (396,272) -- cycle ;
	\draw   (140,218) -- (513,218) -- (513,283) -- (140,283) -- cycle ;
	\draw  [fill={rgb, 255:red, 139; green, 87; blue, 42 }  ,fill opacity=0.4 ] (217,231) -- (261,231) -- (261,271) -- (217,271) -- cycle ;
	\draw   (159,144) -- (203,144) -- (203,184) -- (159,184) -- cycle ;
	\draw   (277,144) -- (321,144) -- (321,184) -- (277,184) -- cycle ;
	\draw   (397,144) -- (441,144) -- (441,184) -- (397,184) -- cycle ;
	\draw   (455,144) -- (499,144) -- (499,184) -- (455,184) -- cycle ;
	\draw   (141,130) -- (514,130) -- (514,195) -- (141,195) -- cycle ;
	\draw   (217,144) -- (261,144) -- (261,184) -- (217,184) -- cycle ;
	\draw  [fill={rgb, 255:red, 208; green, 2; blue, 27 }  ,fill opacity=0.4 ] (339,144) -- (383,144) -- (383,184) -- (339,184) -- cycle ;
	\draw    (241.44,249.26) -- (358.56,165.74) ;
	\draw [shift={(361,164)}, rotate = 144.51] [fill={rgb, 255:red, 0; green, 0; blue, 0 }  ][line width=0.08]  [draw opacity=0] (10.72,-5.15) -- (0,0) -- (10.72,5.15) -- (7.12,0) -- cycle    ;
	\draw [shift={(239,251)}, rotate = 324.51] [fill={rgb, 255:red, 0; green, 0; blue, 0 }  ][line width=0.08]  [draw opacity=0] (10.72,-5.15) -- (0,0) -- (10.72,5.15) -- (7.12,0) -- cycle    ;
	\draw   (337,55) -- (381,55) -- (381,95) -- (337,95) -- cycle ;
	\draw   (397,55) -- (441,55) -- (441,95) -- (397,95) -- cycle ;
	\draw   (141,41) -- (514,41) -- (514,106) -- (141,106) -- cycle ;
	\draw   (218,55) -- (262,55) -- (262,95) -- (218,95) -- cycle ;
	\draw  [fill={rgb, 255:red, 248; green, 231; blue, 28 }  ,fill opacity=0.4 ] (457,53) -- (501,53) -- (501,93) -- (457,93) -- cycle ;
	\draw  [fill={rgb, 255:red, 74; green, 144; blue, 226 }  ,fill opacity=0.4 ] (459,232) -- (503,232) -- (503,272) -- (459,272) -- cycle ;
	\draw    (491.18,251.93) .. controls (594.59,215.95) and (599.78,106.82) .. (485.73,78.42) ;
	\draw [shift={(484,78)}, rotate = 13.46] [fill={rgb, 255:red, 0; green, 0; blue, 0 }  ][line width=0.08]  [draw opacity=0] (10.72,-5.15) -- (0,0) -- (10.72,5.15) -- (7.12,0) -- cycle    ;
	\draw [shift={(488,253)}, rotate = 341.89] [fill={rgb, 255:red, 0; green, 0; blue, 0 }  ][line width=0.08]  [draw opacity=0] (10.72,-5.15) -- (0,0) -- (10.72,5.15) -- (7.12,0) -- cycle    ;
	\draw  [fill={rgb, 255:red, 189; green, 16; blue, 224 }  ,fill opacity=0.4 ] (161,55) -- (205,55) -- (205,95) -- (161,95) -- cycle ;
	\draw  [fill={rgb, 255:red, 126; green, 211; blue, 33 }  ,fill opacity=0.4 ] (279,54) -- (323,54) -- (323,94) -- (279,94) -- cycle ;
	\draw    (301.87,70.97) .. controls (318.28,6.08) and (173.14,-12.3) .. (182.66,72.4) ;
	\draw [shift={(183,75)}, rotate = 261.69] [fill={rgb, 255:red, 0; green, 0; blue, 0 }  ][line width=0.08]  [draw opacity=0] (10.72,-5.15) -- (0,0) -- (10.72,5.15) -- (7.12,0) -- cycle    ;
	\draw [shift={(301,74)}, rotate = 287.93] [fill={rgb, 255:red, 0; green, 0; blue, 0 }  ][line width=0.08]  [draw opacity=0] (10.72,-5.15) -- (0,0) -- (10.72,5.15) -- (7.12,0) -- cycle    ;
	\draw    (40,311) -- (632,311) ;
	
	\draw (721,21) node    {$$};
	\draw (701,71) node    {$$};

	\end{tikzpicture}}
	\caption{Each rectangle illustrates an agent and the squares inside the rectangles represent the items given to the agents. The agent below the line is easily satisfiable and the other agents are not-easily satisfiable. Except for the white squares, the squares that have the same color are the same items.}\label{figure:combine}
\end{figure}
Finally, we introduce a correlated randomized allocation technique that turns the 2-multiallocation into an allocation that satisfies the condition of \Cref{lemma:second}. The idea behind this correlated randomized allocation is shown in \Cref{figure:combine}. Notice that each item which is shared by more than one agent is shared by an easily-satisfiable and a not-easily-satisfiable agent. As aforementioned, for each easily-satisfiable agent, we introduce an auxiliary linear function for the items allocated to it. For an easily-satisfiable agent $\hat{\agent}_i$ we sort the allocated items to it based on her auxiliary valuation function. We then pair her allocated items as follows: we pair the first two items (the ones that have the highest utility in her auxiliary function), we then pair the next two items and so on. If the number of items allocated to that agent is odd, we add a dummy item with value 0 to her bundle. We run the same procedure for all easily-satisfiable agents.

In order to turn our 2-multiallocation into a valid allocation we run the following procedure:
\begin{itemize}
	\item If an item is given to a single agent in the 2-multiallocation it will be given to the same agent in the allocation as well.
	\item For each two items $(\hat{\ite}_x, \hat{\ite}_y)$ that are paired and are given to different not-easily-satisfiable agents, we flip a coin. If the outcome is heads, we give item $\hat{\ite}_x$ to the easily-satisfiable agent who received it in the 2-multiallocation and we give item $\hat{\ite}_y$ to the not-easily-satisfiable agent that received $\hat{\ite}_y$ in the 2-multiallocation. Otherwise, we give item $\hat{\ite}_x$ to the not-easily-satisfiable agent who received it in the 2-multiallocation and we give item $\hat{\ite}_y$ to the easily-satisfiable agent that received $\hat{\ite}_y$ in the 2-multiallocation.
	\item Up to this point, the only items that are not allocated yet are the paired items that are given to the same not-easily-satisfiable agents in the 2-multiallocation. In this step, each not-easily-satisfiable agent looks at all  such paired items and from each pair chooses one item in a way that their combination maximize her utility. The items that are not chosen by the not-easily-satisfiable agents will go to the easily-satisfiable agents that own them in the 2-multiallocation.
\end{itemize} 
We show in the proof of Lemma~\ref{lemma:second} that such a procedure guarantees our desired conditions with non-zero probability (and thus there exists an allocation that satisfies the conditions of Lemma~\ref{lemma:second}).

\subsection{Approximate Solutions}\label{sec:o2}
We begin by stating a rather simple yet very important corollary of Lemma~\ref{lemma:second}. This provides a direct reduction from the \MMS\ problem to a relaxed version of the \MMS\ problem wherein items are allowed to be given to multiple agents.

\vspace{0.2cm}
{\noindent \textbf{Lemma} \ref{lemma:reduction} [restated informally, a corollary of Lemma~\ref{lemma:second}]. \textit{The existence of an $\alpha$-multiallocation that guarantees a $1/\eta$-\MMS\ approximation for the subadditive maximin share problem leads to the existence of a $1/(c \alpha \eta  (\log \alpha+\log \log n))$-\MMS\ guarantee for the subadditive maximin share problem for some constant $c = O(1)$.\\}}

In what follows, we present several methods to show the existence of multiallocations with provable \MMS\ guarantees.

\subsubsection{Warm-up 1: a $1/O(\log n \log \log n)$-\MMS\ guarantee}\label{section:trivial}
To show the effectiveness of Lemma~\ref{lemma:second}, we state a simple method that leads to a $1/O(\log n \log \log n)$-\MMS\ guarantees for the subadditive maximine share problem. Recall that this result has already been presented by \textcite{maseed123}. However, here we simplify the proof using Lemma~\ref{lemma:second}.

It is proven in~\cite{maseed123}, that for any subset $Q \subseteq N$ of agents, there exists an allocation of items to the agents of $Q$ such that at least a constant fraction of agents in $Q$ receive a bundle whose value to them is at least a constant fraction of their \MMS\ value. Therefore we can construct an $O(\log n)$-multiallocation that guarantees a constant \MMS\ guarantee for the agents in the following way: We start by $Q = N$ and find an allocation of items to agents of $Q$ in a way that a constant fraction of the agents in $Q$ receive a bundle whose value to them is a constant fraction of their \MMS\ value. We then update $Q$ by removing agents whose allocated bundles provide a constant fraction of their \MMS\ values and repeat the same procedure for the rest of the agents. Note that since each time the size of $Q$ is multiplied by a constant factor, the algorithm terminates after $O(\log n)$ iterations and thus this leads to an $O(\log n)$-multiallocation that guarantees a constant fraction of the \MMS\ value. This in addition to the reduction presented in Lemma~\ref{lemma:reduction} leads to a $1/O(\log n  \log \log n)$-\MMS\ guarantee for the maximin share problem with subadditive agents.

\vspace{0.2cm}
{\noindent \textbf{Theorem} \ref{theorem:w1} [restated informally]. \textit{The maximin share problem with subadditive agents admits a $1/O(\log n  \log \log n)$-\MMS\ guarantee.\\}}

\subsubsection{Warm-up 2: Improving to Sublogarithmic for Polynomial $m$}
We are now ready to present our first idea to go beyond logarithmic guarantees. To this end, we prove a more advanced bound on the partial allocation of items to agents. We show in Lemma~\ref{lemma:third} that for any subset $Q \subseteq N$ there exist multiple disjoint allocations to agents such that at least a constant fraction of the agents of $Q$ receive a bundle whose value to them is a constant fraction of their \MMS\ values.

\vspace{0.2cm}
{\noindent \textbf{Lemma} \ref{lemma:third} [restated informally]. \textit{Let $Q \subseteq \agents$ be a subset of agents. For $k = |\agents| / (6|Q|) $, there exists a subset $Q' \subseteq Q$ and $\lceil k \rceil$ disjoint allocations  $\mathcal{\alloc}^1, \mathcal{\alloc}^2,\ldots, \mathcal{\alloc}^{\lceil k \rceil}$ of items to agents of $Q'$ such that  $|Q'| = \Omega(|Q|)$ and for every agent $a_{i} \in Q'$ and $1 \leq j \leq \lceil k \rceil$ we have 
		$\valu_{i}(A^j_{i}) \geq c$ for some  $c = \Omega(1)$.\\}}

We bring the proof of Lemma~\ref{lemma:third} in \Cref{sec:sublogarithmicapprox}. Here we show how this result can lead to an improved bound on \MMS\ allocations. Similar to what we explained earlier, we start by $Q = N$ and iteratively construct many disjoint allocations for a constant fraction of agents in $Q$ and remove them from $Q$. The key difference to what we did previously is that this time, instead of a single allocation we give them multiple disjoint allocations due to Lemma~\ref{lemma:third}. After the termination of the algorithm in $O(\log n)$ steps, we construct a multiallocation by randomly giving one of the bundles allocated to each agent in any of the $O(\log n)$ rounds of the algorithm. Although each item may be given to $O(\log n)$ different agents throughout the procedure, we show that with non-zero probability, no item appears in more than $O(\sqrt{\log |M|})$ randomly chosen bundles. This married with Lemma~\ref{lemma:reduction} yields an improved bound for the cases where the number of items is polynomial.

\vspace{0.2cm}
{\noindent \textbf{Theorem} \ref{theorem:w2} [restated informally]. \textit{The maximin share problem with subadditive agents admits a $1/O(\sqrt{\log m} \log \log m)$-\MMS\ guarantee.\\}}

\subsubsection{Main Contribution: A $1/O((\log \log n)^2)$-\MMS\ Guarantee}
Finally, we are ready to present our main contribution. We show that the maximin share problem with subadditive agents admits a $\nicefrac{1}{O((\log \log n)^2)}$-\MMS\ guarantee. The blueprint of the proof is much like what we explained in Section~\ref{section:trivial}. However, here we prove a more efficient partial allocation lemma and as a result, we create our multiallocation in $O(\log \log n)$ rounds instead of $O(\log n)$ rounds. 

\vspace{0.2cm}
{\noindent \textbf{Lemma} \ref{lemma:fourth} [restated informally]. \textit{Let $Q \subseteq N$ be a subset of agents. There exists a subset $Q' = \{a_{x_1}, a_{x_2}, \ldots,a_{x_{|Q'|}}\} \subseteq Q$ and an allocation  $A_{x_1},A_{x_2},\ldots,A_{x_{|Q'|}}$ of items to agents of $Q'$ such that $|Q'| \geq |Q|\frac{k}{k+1}$ for $k = \lfloor n/|Q| \rfloor$ and for every $a_{x_i} \in Q'$  we have 
		$\valu_{x_i}(A_{x_i}) \geq c$ for some  $c = \Omega(1)$.}}

Leveraging Lemma~\ref{lemma:fourth}, one can create a $\nicefrac{1}{O(\log \log n)}$-multiallocation in which every agent receives a bundle whose value to her is a constant fraction of her \MMS\ value. This is basically similar to what we explained in Section~\ref{section:trivial} except that due to the guarantee of Lemma~\ref{lemma:fourth}, our algorithm terminates after $O(\log \log n)$ iterations. This married with Lemma~\ref{lemma:reduction} gives us an improved bound of $\nicefrac{1}{O((\log \log n)^2)}$-\MMS\ guarantee.

\vspace{0.2cm}
{\noindent \textbf{Theorem} \ref{theorem:main} [restated informally]. \textit{The maximin share problem with subadditive agents admits a $\nicefrac{1}{O((\log \log n)^2)}$-\MMS\ guarantee.\\}}

The challenging part of the analysis however is the proof of Lemma~\ref{lemma:fourth}. This is in fact our deepest technical contribution and we believe it will find its application in future work as well. Our method is based on the seminal work of~\textcite{feige2009maximizing} in which the author presents a 2-approximation algorithm for revenue maximization of subadditive agents. We build on the rounding technique that Feige uses to prove a bound of 2 on the integrality gap of the configuration LP. While we generalize the method to prove a better bound for our case, we believe we also make it more intuitive and also applicable to other scenarios. As we discuss later, the new \textit{guiding graph} and \textit{matching technique} that we introduce to the rounding technique of~\textcite{feige2009maximizing} are important and necessary parts of the analysis. However, since the goal of~\cite{feige2009maximizing} is to show a weaker statement, their guiding graph is made simpler and instead of the matching technique they use an edge orientation technique both of which only work for their special need. As a result, while mathematically correct, there is little intuition as to why the combination of edge orientation and their version of guiding graph leads to the desired outcome. As we show later, both their guiding graph and edge orientation technique can be thought of as special cases of our guiding graph and matching algorithm and thus we believe our generalization also adds intuition and simplifies the technique. Below we outline our technique.

\begin{figure}

\tikzset{every picture/.style={line width=0.75pt}} 

\begin{center}

\scalebox{1}{
\begin{tikzpicture}[x=0.75pt,y=0.75pt,yscale=0.45,xscale=0.45]

\draw   (17,332) -- (642,332) -- (642,397) -- (17,397) -- cycle ;
\draw  [fill={rgb, 255:red, 0; green, 0; blue, 0 }  ,fill opacity=1 ] (38,367) .. controls (38,363.69) and (40.69,361) .. (44,361) .. controls (47.31,361) and (50,363.69) .. (50,367) .. controls (50,370.31) and (47.31,373) .. (44,373) .. controls (40.69,373) and (38,370.31) .. (38,367) -- cycle ;
\draw  [fill={rgb, 255:red, 0; green, 0; blue, 0 }  ,fill opacity=1 ] (68,367) .. controls (68,363.69) and (70.69,361) .. (74,361) .. controls (77.31,361) and (80,363.69) .. (80,367) .. controls (80,370.31) and (77.31,373) .. (74,373) .. controls (70.69,373) and (68,370.31) .. (68,367) -- cycle ;
\draw  [fill={rgb, 255:red, 0; green, 0; blue, 0 }  ,fill opacity=1 ] (98,367) .. controls (98,363.69) and (100.69,361) .. (104,361) .. controls (107.31,361) and (110,363.69) .. (110,367) .. controls (110,370.31) and (107.31,373) .. (104,373) .. controls (100.69,373) and (98,370.31) .. (98,367) -- cycle ;
\draw  [fill={rgb, 255:red, 0; green, 0; blue, 0 }  ,fill opacity=1 ] (128,367) .. controls (128,363.69) and (130.69,361) .. (134,361) .. controls (137.31,361) and (140,363.69) .. (140,367) .. controls (140,370.31) and (137.31,373) .. (134,373) .. controls (130.69,373) and (128,370.31) .. (128,367) -- cycle ;
\draw  [fill={rgb, 255:red, 0; green, 0; blue, 0 }  ,fill opacity=1 ] (158,367) .. controls (158,363.69) and (160.69,361) .. (164,361) .. controls (167.31,361) and (170,363.69) .. (170,367) .. controls (170,370.31) and (167.31,373) .. (164,373) .. controls (160.69,373) and (158,370.31) .. (158,367) -- cycle ;
\draw  [fill={rgb, 255:red, 0; green, 0; blue, 0 }  ,fill opacity=1 ] (188,367) .. controls (188,363.69) and (190.69,361) .. (194,361) .. controls (197.31,361) and (200,363.69) .. (200,367) .. controls (200,370.31) and (197.31,373) .. (194,373) .. controls (190.69,373) and (188,370.31) .. (188,367) -- cycle ;
\draw  [fill={rgb, 255:red, 0; green, 0; blue, 0 }  ,fill opacity=1 ] (220,367) .. controls (220,363.69) and (222.69,361) .. (226,361) .. controls (229.31,361) and (232,363.69) .. (232,367) .. controls (232,370.31) and (229.31,373) .. (226,373) .. controls (222.69,373) and (220,370.31) .. (220,367) -- cycle ;
\draw  [fill={rgb, 255:red, 0; green, 0; blue, 0 }  ,fill opacity=1 ] (250,367) .. controls (250,363.69) and (252.69,361) .. (256,361) .. controls (259.31,361) and (262,363.69) .. (262,367) .. controls (262,370.31) and (259.31,373) .. (256,373) .. controls (252.69,373) and (250,370.31) .. (250,367) -- cycle ;
\draw  [fill={rgb, 255:red, 0; green, 0; blue, 0 }  ,fill opacity=1 ] (280,367) .. controls (280,363.69) and (282.69,361) .. (286,361) .. controls (289.31,361) and (292,363.69) .. (292,367) .. controls (292,370.31) and (289.31,373) .. (286,373) .. controls (282.69,373) and (280,370.31) .. (280,367) -- cycle ;
\draw  [fill={rgb, 255:red, 0; green, 0; blue, 0 }  ,fill opacity=1 ] (378,367) .. controls (378,363.69) and (380.69,361) .. (384,361) .. controls (387.31,361) and (390,363.69) .. (390,367) .. controls (390,370.31) and (387.31,373) .. (384,373) .. controls (380.69,373) and (378,370.31) .. (378,367) -- cycle ;
\draw  [fill={rgb, 255:red, 0; green, 0; blue, 0 }  ,fill opacity=1 ] (408,367) .. controls (408,363.69) and (410.69,361) .. (414,361) .. controls (417.31,361) and (420,363.69) .. (420,367) .. controls (420,370.31) and (417.31,373) .. (414,373) .. controls (410.69,373) and (408,370.31) .. (408,367) -- cycle ;
\draw  [fill={rgb, 255:red, 0; green, 0; blue, 0 }  ,fill opacity=1 ] (438,367) .. controls (438,363.69) and (440.69,361) .. (444,361) .. controls (447.31,361) and (450,363.69) .. (450,367) .. controls (450,370.31) and (447.31,373) .. (444,373) .. controls (440.69,373) and (438,370.31) .. (438,367) -- cycle ;
\draw  [fill={rgb, 255:red, 0; green, 0; blue, 0 }  ,fill opacity=1 ] (468,367) .. controls (468,363.69) and (470.69,361) .. (474,361) .. controls (477.31,361) and (480,363.69) .. (480,367) .. controls (480,370.31) and (477.31,373) .. (474,373) .. controls (470.69,373) and (468,370.31) .. (468,367) -- cycle ;
\draw  [fill={rgb, 255:red, 0; green, 0; blue, 0 }  ,fill opacity=1 ] (498,367) .. controls (498,363.69) and (500.69,361) .. (504,361) .. controls (507.31,361) and (510,363.69) .. (510,367) .. controls (510,370.31) and (507.31,373) .. (504,373) .. controls (500.69,373) and (498,370.31) .. (498,367) -- cycle ;
\draw  [fill={rgb, 255:red, 0; green, 0; blue, 0 }  ,fill opacity=1 ] (528,367) .. controls (528,363.69) and (530.69,361) .. (534,361) .. controls (537.31,361) and (540,363.69) .. (540,367) .. controls (540,370.31) and (537.31,373) .. (534,373) .. controls (530.69,373) and (528,370.31) .. (528,367) -- cycle ;
\draw  [fill={rgb, 255:red, 0; green, 0; blue, 0 }  ,fill opacity=1 ] (560,367) .. controls (560,363.69) and (562.69,361) .. (566,361) .. controls (569.31,361) and (572,363.69) .. (572,367) .. controls (572,370.31) and (569.31,373) .. (566,373) .. controls (562.69,373) and (560,370.31) .. (560,367) -- cycle ;
\draw  [fill={rgb, 255:red, 0; green, 0; blue, 0 }  ,fill opacity=1 ] (590,367) .. controls (590,363.69) and (592.69,361) .. (596,361) .. controls (599.31,361) and (602,363.69) .. (602,367) .. controls (602,370.31) and (599.31,373) .. (596,373) .. controls (592.69,373) and (590,370.31) .. (590,367) -- cycle ;
\draw  [fill={rgb, 255:red, 0; green, 0; blue, 0 }  ,fill opacity=1 ] (620,367) .. controls (620,363.69) and (622.69,361) .. (626,361) .. controls (629.31,361) and (632,363.69) .. (632,367) .. controls (632,370.31) and (629.31,373) .. (626,373) .. controls (622.69,373) and (620,370.31) .. (620,367) -- cycle ;
\draw   (20,58) -- (143,58) -- (143,123) -- (20,123) -- cycle ;
\draw  [fill={rgb, 255:red, 0; green, 0; blue, 0 }  ,fill opacity=1 ] (36,91) .. controls (36,87.69) and (38.69,85) .. (42,85) .. controls (45.31,85) and (48,87.69) .. (48,91) .. controls (48,94.31) and (45.31,97) .. (42,97) .. controls (38.69,97) and (36,94.31) .. (36,91) -- cycle ;
\draw  [fill={rgb, 255:red, 0; green, 0; blue, 0 }  ,fill opacity=1 ] (116,91) .. controls (116,87.69) and (118.69,85) .. (122,85) .. controls (125.31,85) and (128,87.69) .. (128,91) .. controls (128,94.31) and (125.31,97) .. (122,97) .. controls (118.69,97) and (116,94.31) .. (116,91) -- cycle ;
\draw   (210,58) -- (333,58) -- (333,123) -- (210,123) -- cycle ;
\draw  [fill={rgb, 255:red, 0; green, 0; blue, 0 }  ,fill opacity=1 ] (226,91) .. controls (226,87.69) and (228.69,85) .. (232,85) .. controls (235.31,85) and (238,87.69) .. (238,91) .. controls (238,94.31) and (235.31,97) .. (232,97) .. controls (228.69,97) and (226,94.31) .. (226,91) -- cycle ;
\draw  [fill={rgb, 255:red, 0; green, 0; blue, 0 }  ,fill opacity=1 ] (306,91) .. controls (306,87.69) and (308.69,85) .. (312,85) .. controls (315.31,85) and (318,87.69) .. (318,91) .. controls (318,94.31) and (315.31,97) .. (312,97) .. controls (308.69,97) and (306,94.31) .. (306,91) -- cycle ;
\draw   (510,55) -- (633,55) -- (633,120) -- (510,120) -- cycle ;
\draw  [fill={rgb, 255:red, 0; green, 0; blue, 0 }  ,fill opacity=1 ] (526,88) .. controls (526,84.69) and (528.69,82) .. (532,82) .. controls (535.31,82) and (538,84.69) .. (538,88) .. controls (538,91.31) and (535.31,94) .. (532,94) .. controls (528.69,94) and (526,91.31) .. (526,88) -- cycle ;
\draw  [fill={rgb, 255:red, 0; green, 0; blue, 0 }  ,fill opacity=1 ] (606,88) .. controls (606,84.69) and (608.69,82) .. (612,82) .. controls (615.31,82) and (618,84.69) .. (618,88) .. controls (618,91.31) and (615.31,94) .. (612,94) .. controls (608.69,94) and (606,91.31) .. (606,88) -- cycle ;
\draw    (122,91) -- (162,367) ;
\draw [shift={(162,367)}, rotate = 81.75] [color={rgb, 255:red, 0; green, 0; blue, 0 }  ][fill={rgb, 255:red, 0; green, 0; blue, 0 }  ][line width=0.75]      (0, 0) circle [x radius= 3.35, y radius= 3.35]   ;
\draw [shift={(122,91)}, rotate = 81.75] [color={rgb, 255:red, 0; green, 0; blue, 0 }  ][fill={rgb, 255:red, 0; green, 0; blue, 0 }  ][line width=0.75]      (0, 0) circle [x radius= 3.35, y radius= 3.35]   ;
\draw    (312,91) -- (164,373) ;
\draw [shift={(164,373)}, rotate = 117.69] [color={rgb, 255:red, 0; green, 0; blue, 0 }  ][fill={rgb, 255:red, 0; green, 0; blue, 0 }  ][line width=0.75]      (0, 0) circle [x radius= 3.35, y radius= 3.35]   ;
\draw [shift={(312,91)}, rotate = 117.69] [color={rgb, 255:red, 0; green, 0; blue, 0 }  ][fill={rgb, 255:red, 0; green, 0; blue, 0 }  ][line width=0.75]      (0, 0) circle [x radius= 3.35, y radius= 3.35]   ;
\draw    (532,88) -- (164,373) ;
\draw [shift={(164,373)}, rotate = 142.24] [color={rgb, 255:red, 0; green, 0; blue, 0 }  ][fill={rgb, 255:red, 0; green, 0; blue, 0 }  ][line width=0.75]      (0, 0) circle [x radius= 3.35, y radius= 3.35]   ;
\draw [shift={(532,88)}, rotate = 142.24] [color={rgb, 255:red, 0; green, 0; blue, 0 }  ][fill={rgb, 255:red, 0; green, 0; blue, 0 }  ][line width=0.75]      (0, 0) circle [x radius= 3.35, y radius= 3.35]   ;
\draw    (122,97) -- (44,373) ;
\draw [shift={(44,373)}, rotate = 105.78] [color={rgb, 255:red, 0; green, 0; blue, 0 }  ][fill={rgb, 255:red, 0; green, 0; blue, 0 }  ][line width=0.75]      (0, 0) circle [x radius= 3.35, y radius= 3.35]   ;
\draw [shift={(122,97)}, rotate = 105.78] [color={rgb, 255:red, 0; green, 0; blue, 0 }  ][fill={rgb, 255:red, 0; green, 0; blue, 0 }  ][line width=0.75]      (0, 0) circle [x radius= 3.35, y radius= 3.35]   ;
\draw    (122,91) -- (504,367) ;
\draw [shift={(504,367)}, rotate = 35.85] [color={rgb, 255:red, 0; green, 0; blue, 0 }  ][fill={rgb, 255:red, 0; green, 0; blue, 0 }  ][line width=0.75]      (0, 0) circle [x radius= 3.35, y radius= 3.35]   ;
\draw [shift={(122,91)}, rotate = 35.85] [color={rgb, 255:red, 0; green, 0; blue, 0 }  ][fill={rgb, 255:red, 0; green, 0; blue, 0 }  ][line width=0.75]      (0, 0) circle [x radius= 3.35, y radius= 3.35]   ;
\draw    (122,91) -- (596,367) ;
\draw [shift={(596,367)}, rotate = 30.21] [color={rgb, 255:red, 0; green, 0; blue, 0 }  ][fill={rgb, 255:red, 0; green, 0; blue, 0 }  ][line width=0.75]      (0, 0) circle [x radius= 3.35, y radius= 3.35]   ;
\draw [shift={(122,91)}, rotate = 30.21] [color={rgb, 255:red, 0; green, 0; blue, 0 }  ][fill={rgb, 255:red, 0; green, 0; blue, 0 }  ][line width=0.75]      (0, 0) circle [x radius= 3.35, y radius= 3.35]   ;

\draw (721,21) node    {$$};
\draw (701,71) node    {$$};
\draw (313,365.4) node [anchor=north west][inner sep=0.75pt]  [font=\LARGE]  {$\cdots $};
\draw (57,90.4) node [anchor=north west][inner sep=0.75pt]  [font=\LARGE]  {$\dotsc $};
\draw (247,90.4) node [anchor=north west][inner sep=0.75pt]  [font=\LARGE]  {$\dotsc $};
\draw (547,90.4) node [anchor=north west][inner sep=0.75pt]  [font=\LARGE]  {$\dotsc $};
\draw (395,90.4) node [anchor=north west][inner sep=0.75pt]  [font=\LARGE]  {$\dotsc $};

\end{tikzpicture}}
\end{center}
	\caption{Bottom nodes represent $|Q|$ parts of the allocation nodes and top nodes represent the random-seed nodes. Each random-seed node has exactly one edge to each part of the allocation nodes. Each allocation nodes has exactly $k+1$ edges to the random-seed nodes.}\label{figure:guiding-graph}
\end{figure}
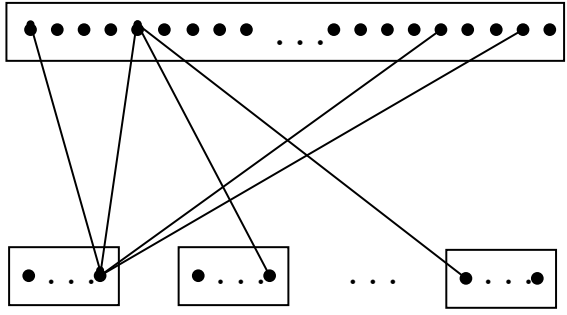

Recall that in Lemma~\ref{lemma:fourth} we have an \MMS\ problem in which every agent has $n$ disjoint choices, each of which provides a value of at least 1 to her. However, we only wish to allocate items to a smaller subset of agents $Q$. Recall that we denote the ratio of $n$ over $|Q|$ by $k = \lfloor n/|Q| \rfloor$. The goal is to prove that there exists an allocation that provides a value of at least $1/2$ to at least a $\frac{k}{k+1}$ fraction of agents of $Q$. To prove such an allocation exists, we present a randomized process. We define a guiding graph $\mathcal{G}$ in the following way:
\begin{itemize}
	\item $|\mathcal{G}|$ is of no importance in our analysis, but it has exponentially many vertices and edges.
	\item $\mathcal{G}$ has a large enough girth.
	\item $\mathcal{G}$ is a bipartite graph. We call the vertices of one part \textit{random-seed nodes}  and the vertices of the other part \textit{allocation nodes}.
	\item Each random-seed node is basically one possible outcome of our process. More precisely, in the last step of our allocation, we select one of the random-seed nodes randomly and by looking at the neighbors of that node, we determine which items are allocated to each agent.
	\item Allocation nodes are divided into $|Q|$ parts, each corresponding to one agent. The degree of each allocation node is exactly $k+1$.
	\item Every random-seed node has exactly $|Q|$ edges each corresponding to one agent of $|Q|$. More precisely, each edge is connected to one part of the allocation nodes.
\end{itemize}

Apart from the structure of the guiding graph, the only feature which makes it usable for our purpose is that it has a high girth. We show in Section~\ref{sec:main} that for any girth $g$, there exists a graph with the above structure that has a girth of at least $g$. Notice that each of the allocation nodes corresponds to one agent of $Q$. In the first step of our randomized process, we label each allocation node with one of the $n$ bundles of its corresponding agent whose value to her is at least $1$ uniformly at random. Keep in mind that based on this, each item appears in the label of an allocation node with probability $1/n$. In our randomized process, we plan to select one of the random-seed nodes uniformly at random and by looking at the label of its neighbors determine which items are allocated to each agent. However, at this point, this is not possible since the labels of the neighbors of the random-seed nodes are not disjoint. In other words, if we select one random-seed node and allocate the items based on its neighbors, an item may be given to multiple agents. We resolve this issue by making the following modification:
\begin{itemize}
	\item For each allocation node, we label each of its edges by a subset of its label.
	\item We maintain the following two properties in our labelling: (i) The labels of the edges of each random-seed vertex will be disjoint. (ii) for each allocation node, none of the items of its label is missing in the label of more than one of its edges. In other words, each item in the label of an allocation node is present in the labels of at least $k$ (out of $k+1$) of its edges.
\end{itemize}
While technically such a labelling may not be possible in general, we show in the following that by ignoring some minor details we can assume for the purpose of our discussion that such a labelling exists. The correct statement that we show in Section~\ref{sec:main} is that there exists a labelling that satisfies the above properties for almost all of the nodes (except very few nodes that can be ignored). Let us fix an item $\ite_x$ and show how we decide which edge labels of the guiding graph contain $\ite_x$. To this end, consider an induced subgraph $\mathcal{H}_x$ of $\mathcal{G}$ wherein allocation nodes whose labels do not contain $\ite_x$ are removed. While this is not technically correct, we assume for simplicity that $\mathcal{H}_x$ is a forest. 

Let us first explain why such an assumption is almost without loss of generality. Notice that the expected degree of each random-seed node of $\mathcal{H}_x$ is at most $1/k$. On the other hand, the degree of each allocation node of $\mathcal{H}_x$ is exactly $k+1$. We show in Section~\ref{sec:main} that this implies an upper bound on the expected size of the connected components of $\mathcal{H}_x$. Given that the girth of $\mathcal{G}$ is large enough, we also know that the girth of $\mathcal{H}_x$ is large enough and thus  almost all vertices of $\mathcal{H}_x$ are in connected components that are trees. For these reasons, we assume for simplicity that $\mathcal{H}_x$ is a forest here though we bring detailed discussions in Section~\ref{sec:main}.

Our aim here is to determine $\ite_x$ appears in the label of which edges of  $\mathcal{H}_x$. However, we have two restrictions: (i) For each random-seed node of $\mathcal{H}_x$ at most the label of one of its edges can have $\ite_x$. (ii) for each allocation node of $\mathcal{H}_x$, $\ite_x$ should appear in at least $k$ (out of $k+1$) labels of its edges. To prove that this is possible, we use the Hall theorem~\cite{cormen2022introduction}. Let $Y$ be a subset of the allocation nodes of $\mathcal{H}_x$. The number of edges incident to the vertices of $Y$ is exactly $(k+1)|Y|$. Since $\mathcal{H}_x$ is a forest, there are at least $(k+1)|Y|+1$ distinct endpoints for these edges. $|Y|$ of these endpoints are the vertices of $Y$ and thus at least $k|Y|+1$ of these endpoints are random-seed nodes. This basically means that the number of neighbors of $Y$ in the random-seed part is at least $k$ times the size of $|Y|$. This means that we can find $k$ disjoint matchings of $\mathcal{H}_x$ each of which covers the allocation nodes. On the other hand, each random-seed vertex of $\mathcal{H}_x$ appears in at most one of these matchings. Therefore, if we include $\ite_x$ in the labels of the edges of these matchings both our conditions are satisfied.

We use the same technique to decide whether each item appears in the labels of each edge of $\mathcal{G}$. Finally, we randomly choose one random-seed node of $\mathcal{G}$ and based on the labels of the edges incident to it determine which items are allocated to each agent of $Q$. We then prove that in expectation, at least a $\frac{k}{k+1}$ fraction of agents receive a bundle which is worth at least $1/2$ to them. This basically proves that there exists an allocation that gives a bundle to each agent of $Q$ such that at least a $\frac{k}{k+1}$ fraction of agents receive an allocation that is worth at least $1/2$ to them. The idea behind the proof is the following: Consider an allocation node of $\mathcal{G}$ and let it be corresponding to agent $\agent_{x_i} \in Q$. Let $T$ be the label of this node. We know that the valuation of agent $\agent_{x_i}$ for items of $T$ is at at least $1$. On the other hand, each item of $T$ is present in the labels of all but at most one edge incident to this node. We show in Section~\ref{sec:main} that this implies that the valuation of agent $\agent_{x_i}$ is at least $1/2$ for at least $k$ (out of $k+1$) labels of the edges incident to this node. This basically means that the probability that an agent receives a bundle that is worth at least $1/2$ to her in this process is at least $\frac{k}{k+1}$ which completes the proof. More details about this proof are included in Section~\ref{sec:main}.

Finally, we remark that for the special case of $k=1$, the degree of each allocation vertex of $\mathcal{G}$ would become 2. Thus, we could think of the allocation vertices of $\mathcal{G}$ as edges that connect their two random-seed neighbors. With such a representation, the resulting graph would be equivalent to the guiding graph of \textcite{feige2009maximizing}. Also, in that scenario, the goal of the matching process would be to find a matching from the edges to the vertices of a tree which would be feasible by making the tree rooted at some arbitrary node and then match each vertex to the edge connecting it to its parent. However, this idea is not generalizable to $k > 1$.

\begin{toappendix}
	\section{Concentration Bound}\label{sec:first}
In this section, we prove Lemma~\ref{lemma:1}
\begin{lemma}\label{lemma:1}
	Let for an $\hat{n} \geq 1$ and a ground set of elements ${\hat{M}}$, $f:2^{\hat{M}} \rightarrow \mathbb{R}_+$ be a monotone subadditive function with non-negative values such that $f(S) \leq f({\hat{M}})/2$ holds for any subset $S \subseteq {\hat{M}}$ such that $|S| \leq \frac{40 (\log \hat{n}+1)}{p}$. Let $R$ be a random subset of ${\hat{M}}$ such that each element of ${\hat{M}}$ appears in $R$ with probability $0 \leq p \leq 1$ independently. Then we have:
	$$\mathsf{Pr} \left[ f(R) \geq \frac{f({\hat{M}})p}{120} \right] > 1-1/\hat{n}.$$
\end{lemma}
The proof of \Cref{lemma:1} is based the bound of~\cite{dobzinski2024constant} for subadditive set functions. Before we state the proof we bring the bound here.

\begin{theorem}[\cite{dobzinski2024constant}, Theorem 20]\label{theorem:them}
	Let $f : 2^{\hat{M}} \to \mathbb{R}^+$ be a monotone subadditive function, where $f(\{\ite_x\}) \leq \nu$, $\nu > 0$,  
	for every $\ite_x \in \hat{M}$. Then for any integers $k, q \geq 1$, and a random set $R$ where elements appear  
	independently,  
	\[
	\mathsf{Pr}\left[ f(R) \leq \frac{\mathbb{E}[f(R)]}{5(q + 1)} - \frac{(k + 1)\nu}{q + 1} \right] 
	\leq \left( \frac{2}{q^k} \right)^{1/q}.
	\]
\end{theorem}
Also, the proof of \Cref{lemma:1} requires a bound on the expectation of randomly selected subsets of a subadditive function which we state here.

\begin{lemma}\label{lemma:exp}
	Let $f : 2^{\hat{M}} \to \mathbb{R}^+$ be a monotone subadditive function with non-negative valuations. If we select a subset $R \subseteq {\hat{M}}$ by independently putting each element into $R$ with probability $p$, we have $\mathbb{E}[f(R)] \geq pf({\hat{M}})/2$.
\end{lemma}
\begin{proof}
	It is shown in~\cite{feige2009maximizing} that when $1/p$ is an integer number, $\mathbb{E}[f(R)] \geq pf({\hat{M}})$ holds. Let $p' = 1/{\lceil 1/p \rceil}$ and $R' \subseteq {\hat{M}}$ be a randomly selected subset of ${\hat{M}}$ by independently putting  each element of ${\hat{M}}$ into $R'$ with probability $p'$. Since $p/2 \leq p' \leq p$ and $f$ is a monotone subadditive function with non-negative valuations, we have: 
	$$p f({\hat{M}})/2 \leq p' f({\hat{M}})  \leq  \mathbb{E}[f(R')] \leq  \mathbb{E}[f(R)].$$
\end{proof}

\begin{proof}[ of Lemma~\ref{lemma:1}]
Recall that a set function $g : 2^{\hat{M}} \to \mathbb{R}^+$ is subadditive if and only if for every subset $X \subseteq {\hat{M}}$ we have:
$$g(X) \leq \sum_{X' \in \langle X_1, X_2, \ldots, X_t \rangle} g(X')$$
for every partitioning of $X$ into non-empty and disjoint subsets $\langle X_1, X_2, \ldots, X_t \rangle$. Based on this, we introduce a bounded function $\bar{f}$ in the following way:

\begin{equation*}
\bar{f}(X) = \begin{cases}
	
	0, & \text{if } |X| = 0 \\
	
	\min\{f(X),f({\hat{M}})/ \frac{80 (\log \hat{n}+1)}{p} \}, & \text{if } |X| = 1 \\
	
	\min\{f(X),\min_{\langle X_1, X_2, \ldots, X_t \rangle \in \textsf{Partitions}(X)}\{\sum_{X' \in \langle X_1, X_2, \ldots, X_t \rangle} \bar{f}(X')\} \}, & \text{if } |X| > 1 \\
	
\end{cases}
\end{equation*}
where $\textsf{Partitions}(X)$ is the set of all partitionings of $X$ into non-empty and disjoint subsets excluding $X$ itself. It follows from the definition that $\bar{f}$ is also a monotone subadditive function and that $f(X) \geq \bar{f}(X)$ for any $X \subseteq {\hat{M}}$. We first show that $\bar{f}({\hat{M}}) \geq f({\hat{M}})/2$.

To this end, assume for the sake of contradiction that $\bar{f}({\hat{M}}) < f({\hat{M}})/2$. Let $\langle X_1, X_2, \ldots, X_t \rangle \in \textsf{Partitions}({\hat{M}})$ be a partitioning of ${\hat{M}}$ that minimizes the expression $\sum_{X' \in \langle X_1, X_2, \ldots, X_t \rangle} \bar{f}(X')$ and has the largest size. In other words $\langle X_1, X_2, \ldots, X_t \rangle \in \textsf{Partitions}({\hat{M}})$ is a partitioning of ${\hat{M}}$ into the maximum number of disjoint subsets such that $\bar{f}({\hat{M}}) = \sum_{X' \in \langle X_1, X_2, \ldots, X_t \rangle} \bar{f}(X')$. This means that for every $X' \in \langle X_1, X_2, \ldots, X_t \rangle$ we have either $|X'| = 1$ or $f(X') = \bar{f}(X')$. We put the partitions of $\langle X_1, X_2, \ldots, X_t \rangle$ into two lists $\mathcal{C} = C_1,C_2, \ldots, C_c$ and $\mathcal{D} = D_1,D_2,\ldots,D_d$ where list $\mathcal{C}$ contains all subsets of $X' \in \langle X_1, X_2, \ldots, X_t \rangle$ for which $f(X') = \bar{f}(X')$ and $\mathcal{D}$ contains the rest of the subsets (this implies that $c+d = t$). Based on this, we know that $|D_i| =  1$ for each $1 \leq i \leq d$ and that $\bar{f}(D_i) =  f({\hat{M}})/\frac{80 (\log \hat{n}+1)}{p}$ for every such subset. Since  we have 
$$f({\hat{M}})/2 > \bar{f}({\hat{M}}) = \sum_{X' \in \langle X_1, X_2, \ldots, X_t \rangle} \bar{f}(X') \geq \sum \bar{f}(D_i) = bf({\hat{M}})/ \frac{80 (\log \hat{n}+1)}{p},$$  it follows that the size of $\mathcal{D}$ is bounded by $\frac{40 (\log \hat{n}+1)}{p}$ (i.e., $d \leq \frac{40 (\log \hat{n}+1)}{p}$). Since $f(S) \leq f({\hat{M}})/2$ holds for any subset $S \subseteq {\hat{M}}$ such that $|S| \leq \frac{40 (\log \hat{n}+1)}{p}$, this implies that $f(D_1 \cup D_2 \cup \ldots \cup D_d) \leq f({\hat{M}})/2$. Moreover, since $\langle C_1, C_2, \ldots, C_c, D_1 \cup D_2 \cup \ldots \cup D_d\rangle$ is a partitioning of ${\hat{M}}$ and $f$ is subadditive, this implies that $\sum_{i=1}^c f(C_i) \geq f({\hat{M}})/2$ and therefore $\sum_{i=1}^c f(C_i) = \sum_{i=1}^c \bar{f}(C_i) \leq \sum_{X' \in \langle X_1, X_2, \ldots, X_t \rangle} \bar{f}(X') = \bar{f}({\hat{M}})$ cannot be smaller than $f({\hat{M}})/2$ which contradicts our assumption.

Notice that with our definition, we have (i) $\bar{f}({\hat{M}}) \geq f({\hat{M}})/2$, (ii) $\bar{f}(X) \leq f(X)$ for every $X \subseteq {\hat{M}}$ and (iii) $\bar{f}(\{\hat{\ite}_x\}) \leq f({\hat{M}}) / \frac{80 (\log \hat{n}+1)}{p}$ for every $\hat{\ite}_x \in {\hat{M}}$. This makes it a great candidate to be applied to Theorem~\ref{theorem:them} to prove our bound. To this end, we choose $k=2\log \hat{n} + 1$, $\nu = f({\hat{M}}) / \frac{80 (\log \hat{n}+1)}{p}$, and $q=2$. By creating $R$ as a random subset of ${\hat{M}}$ whose every element appears in $R$ with probability $p$ we have

	\[
\mathsf{Pr}\left[ \bar{f}(R) \leq \frac{\mathbb{E}[\bar{f}(R)]}{5(q + 1)} - \frac{(k + 1)\nu}{q + 1} \right] 
\leq \left( \frac{2}{q^k} \right)^{1/q}
\]

and therefore 

	\[
\mathsf{Pr}\left[ \bar{f}(R) \leq \frac{\mathbb{E}[\bar{f}(R)]}{5(2 + 1)} - \frac{(2\log \hat{n} + 1 + 1) f({\hat{M}}) / \frac{80 (\log \hat{n}+1)}{p} }{2 + 1} \right] 
\leq \left( \frac{2}{2^{2\log \hat{n} + 1}} \right)^{1/2}
\]

and thus

	\[
\mathsf{Pr}\left[ \bar{f}(R) \leq \frac{\mathbb{E}[\bar{f}(R)]}{15} - \frac{p f({\hat{M}})/40}{3} \right] 
\leq 1/\hat{n}.
\]

Also, it follows from Lemma~\ref{lemma:exp} that $\mathbb{E}[\bar{f}(R)] \geq p\bar{f}({\hat{M}})/2$ since every element of ${\hat{M}}$ appears in $R$ with probability $p$. Therefore we have $\mathbb{E}[\bar{f}(R)] \geq p\bar{f}({\hat{M}})/2 \geq  pf({\hat{M}})/4$. Also keep in mind that $f(R) \geq \bar{f}(R)$ and thus

	\[
\mathsf{Pr}\left[ f(R) \leq \frac{p f({\hat{M}})/4}{15} - \frac{p f({\hat{M}})/40}{3} \right] 
\leq 1/\hat{n}
\]
which implies
	$$\mathsf{Pr} \left[ f(R) > \frac{f({\hat{M}})p}{120} \right] \geq 1-1/\hat{n}.$$
Finally, we note that if we repeat the same analysis by setting $q = 2+\epsilon$ for small enough $\epsilon$, we  obtain 
	$$\mathsf{Pr} \left[ f(R) \geq \frac{f({\hat{M}})p}{120} \right] > 1-1/\hat{n}.$$
\end{proof}

\end{toappendix}

\section{Converting Multiallocation to Allocation}

In this section, we prove Lemma \ref{lemma:second}.

\begin{lemma}\label{lemma:second}
	Let $\tr{\valu}_1, \tr{\valu}_2, \ldots, \tr{\valu}_{\tr{n}}$ be $\tr{n}$ monotone subadditive functions with non-negative valuations defined on a ground set of elements $\tr{M} = \{\tr{\ite}_1, \tr{\ite}_2, \tr{\ite}_3, \ldots, \tr{\ite}_{|\tr{M}|}\}$. Let $A_1, A_2, \ldots, A_{\tr{n}} \subseteq \tr{M}$ be $\tr{n}$ subsets of $\tr{M}$ such that no element of $\tr{M}$ appears in more than $\alpha$ of these subsets. If $\tr{\valu}_{i}({\{\tr{\ite}_x\}}) \leq \beta$ for every $\tr{\ite}_x \in \tr{M}$ and $1 \leq i \leq \tr{n}$ then there exist $\tr{n}$ disjoint subsets $A'_1, A'_2, \ldots, A'_{\tr{n}} \subseteq \tr{M}$ such that:
	\begin{equation}\label[ineq]{eq:lem2}
		\tr{\valu}_i(A'_i) \geq \frac{\tr{\valu}_i(\alloc_i)}{480\alpha(\log (80\alpha) + \log (\log \tr{n}+1))}- \frac{3\beta}{2}.
	\end{equation}
\end{lemma}

Although \Cref{lemma:second} is an independent mathematical result, we present its proof within the context of allocation for brevity and relevance. Specifically, we consider a set  $\tr{\items} = \{\tr{\ite}_1, \tr{\ite}_2, \tr{\ite}_3, \ldots, \tr{\ite}_{|\tr{M}|}\}$ of items, and $\tr{n}$ agents $\tr{\agents} = \{\tr{\agent}_1,\tr{\agent}_2,\ldots, \tr{\agent}_{\tr{n}}\}$ with valuations $\tr{\valu}_1, \tr{\valu}_2, \dots, \tr{\valu}_{\tr{n}}$. $\mathcal{A} = \langle \alloc_1,\alloc_2 \dots, \alloc_{\tr{n}}\rangle$ represents an $\alpha$-multiallocation of $\tr{\items}$. Additionally, we assume that the value of each item for any agent is bounded by $\beta$.

We prove the existence of an allocation satisfying \Cref{eq:lem2} by showing that \Cref{alg:lem2} returns such an allocation. The algorithm begins by dividing the agents into two categories: easily-satisfiable agents ($\easily$) and not-easily-satisfiable agents ($\noteasily$). An agent $\tr{\agent}_i$ is easily-satisfiable if their bundle $A_i$ contains a subset $\subsett_i$ of size at most $80\alpha(\log {\tr n}+1)$ such that $\tr{\valu}_i(\subsett_i) \geq \tr{\valu}_i(A_i)/2$; otherwise, they belong to $\noteasily$. We then produce an intermediary multiallocation based on the bundles of agents in $\easily$ and $\noteasily$ such that the bundles of agents within each category are guaranteed to be disjoint. However, agents from different categories may still share common items. Finally, we run a procedure to produce an allocation $\mathcal{A}'$ based on the intermediary multiallocation with completely disjoint bundles. \Cref{sec:easy,sec:noteasy} describe the construction of the intermediary allocation for easily and not-easily satisfiable agents, respectively.

\begin{algorithm}[t]
	\caption{Multiallocation to Allocation}\label{alg:lem2}
	\LinesNumbered
	\KwIn{Sets $A_1, A_2, \ldots, A_{\tr{n}}$,  $\tr{\valu}_1,\tr{\valu}_2,\ldots, \tr{\valu}_\tr{n}$}
	\KwOut{$A'_1, A'_2,\ldots, A'_{\tr{n}}$}
	$\easily = \{\tr{\agent}_i | \exists X_i \subset A_i \mbox{ such that } |X_i| \leq80\alpha(\log {\tr n} +1) \mbox{ and } \tr{\valu}_i(X_i)\geq\tr{\valu}_i(A_i)/2 \}$\;
	$\noteasily = \{\tr{\agent}_i| \tr{\agent}_i \notin \easily\}$\;
	$\mathcal{\alloc}^\easily$ = \textsc{ResolveEasyIntersections}($\easily$)\;
	$\mathcal{\alloc}^\noteasily$ = \textsc{ResolveHardIntersections}($\noteasily$)\;
	$\mathcal{\alloc}' =  \textsc{Merge}(\mathcal{\alloc}^\easily,\mathcal{\alloc}^\noteasily$)\;
	\ForEach{  $\tr{\agent}_i \in \noteasily$}{
		\If{$\tr{\valu}_i(\alloc'_i) < {\tr{\valu}_i(\alloc_i)}/({480\alpha})$}{Go to Line 4\;}
	}
	Return $\mathcal{\alloc}'$
\end{algorithm}

\subsection{Easily-satisfiable Agents} \label{sec:easy}

The agents in $\easily$ are indeed easier to deal with since their bundles contain a relatively small subset with a high value. For this case, we leverage an idea similar to that of \textcite{ghodsi2018fair} for approximating a subadditive valuation with an additive valuation. Next, we use matching in bipartite graphs to resolve intersections.  For every agent $\tr{\agent_i} \in \easily$, we denote by $A^{\easily}_i$  the subset returned by the algorithm in this step.

Let $\tr{\agent}_i \in \easily$ be an easily-satisfiable agent. By definition, there exists a subset $X_i \subseteq A_i$ such that $\tr{\valu}_i(X_i)\geq\tr{\valu}_i(A_i)/2$ and $|X_i|\leq80 \alpha (\log \tr{n} +1)$. Now, consider the following linear program:

\begin{equation}
	\begin{aligned}
		\text{Maximize} \quad &\sum_{\tr{\ite}_x \in X_i} \widetilde{\valu}_i(\{\tr{\ite}_x\}) \\
		\text{Subject to} \quad & \sum_{\tr{\ite}_x  \in Y} \widetilde{\valu}_i(\{\tr{\ite}_x \}) \leq \tr{\valu}_i(Y) \quad \forall Y \subseteq X_i \\
		& \widetilde{\valu}_i(\{\tr{b}_x \}) \geq 0 \quad \forall \tr{b}_x  \in X_i\\
	\end{aligned}
	\label[LP]{eq:lp}
\end{equation}

In essence, \Cref{eq:lp} seeks to find an additive valuation function $\widetilde{\valu}_i: 2^{\subsett_i} \to \mathbb{R}^+$ that, for every subset of $\subsett_i$, provides a lower-bound approximation of $\tr{\valu}_i$ while maximizing $\widetilde{\valu}_i(X_i)$. As shown by \cite{ghodsi2018fair}, we have $ \widetilde{\valu}_i(X_i)/\tr{\valu}_i(X_i) \geq 1/(3{\log |X_i|})$. Given that $|X_i| \leq 80 \alpha (\log \tr{n} +1)$, it follows that 
\begin{align}
	\frac{\widetilde{\valu}_i(X_i)}{\tr{\valu}_i(X_i)} 	&\geq \frac{1}{3{\log(80 \alpha (\log \tr{n} +1))}} \nonumber \geq \frac{1}{3(\log (80\alpha) + \log (\log \tr{n}+1))}. \numberthis\label[ineq]{delta}
\end{align}
In the rest of the section, for each agent $\tr{\agent}_i \in \easily$, we refer to $\widetilde{\valu}_i$ as the auxiliary valuation of agent $\tr{\agent}_i$. 
We now use bipartite graph matching to construct allocation $\mathcal{A}^\easily$ for the agents in $\easily$. Consider a bipartite graph $G^\easily$ with two parts:  

\begin{enumerate}
	\item The first part includes a node for every item allocated to an agent in $\easily$.
	\item The second part is constructed as follows:  
	For each agent $\tr{\agent}_i \in \easily$, assume the items in $X_i$ are sorted in descending order of their values in $\widetilde{\valu}_i$ as $\tr{b}_{x_1}, \tr{b}_{x_2}, \ldots, \tr{b}_{x_{|X_i|}}$. We group these items into blocks of size $\alpha$ as follows:  
	\begin{align*}
		\block_{i,1} &= [\tr{b}_{x_1}, \tr{b}_{x_2}, \ldots, \tr{b}_{x_\alpha}],  
		\block_{i,2} = [\tr{b}_{x_{\alpha+1}}, \tr{b}_{x_{\alpha+2}}, \ldots, \tr{b}_{x_{2\alpha}}],  
		\ldots  
		,\block_{i,\lfloor |X_i|/\alpha \rfloor} = [\tr{b}_{x_{(\lfloor |X_i|/\alpha \rfloor-1)\alpha+1}}, \ldots, \tr{b}_{x_{\lfloor |X_i|/\alpha \rfloor\alpha}}].
	\end{align*}
	If $|X_i|$ is not divisible by $\alpha$, we ignore items $\tr{b}_{x_{\lfloor |X_i|/\alpha \rfloor\alpha+1}}$ to $\tr{b}_{x_{|X_i|}}$.  For each block created from $X_i$, we add a vertex to the second part of $G^\easily$. Note that, by the way we construct the blocks, for every agent $\tr{\agent}_i \in \easily$ and every $1 \leq j \leq \lfloor |X_i|/\alpha \rfloor-1$ we have
	\begin{equation}\label[ineq]{nblock}
		\forall {\tr{b}_y \in \block_{i,j}, \tr{b}_z \in \block_{i,j+1}} \qquad \widetilde{\valu}_i(\{\tr{b}_y\}) \geq \widetilde{\valu}_i(\{\tr{b}_z\}).
	\end{equation}
\end{enumerate}

We add an edge between the vertex representing an item $\tr{b}_y$ in the first part and the vertex representing a block in the second part if $\tr{b}_y$ appears in that block. Consequently, each vertex corresponding to a block in $G^\easily$ has degree $\alpha$. Since each item appears in at most $\alpha$ bundles, the degree of each vertex in the first part is also at most $\alpha$.

By \emph{Hall}'s theorem \cite{hall1987representatives}, $G^\easily$ has a matching $\matching$ that covers all nodes of the second part. Now, for every $\tr{\agent}_i \in \easily$, we define the bundle of  $\tr{\agent}_i$ in $\mathcal{A}^\easily$  as  
$
A^\easily_i = \{\matching(\block_{i,1}), \matching(\block_{i,2}), \ldots, \matching(\block_{i,\lfloor k/\alpha \rfloor})\},
$  
where $\matching(\block_{i,j})$ is the item corresponding to the vertex matched with the vertex of $\block_{i,j}$ in $\matching$. In \Cref{lem:easily}, we establish a lower bound on the value of the bundles in $ \mathcal{\alloc}^\easily $.

\begin{lemmarep}\label{lem:easily}
	For every agent $\tr{\agent}_i \in \easily$ we have  
	$
	{\widetilde\valu}_i({A}^{\easily}_i) \geq \frac{\tr{\valu}_i(A_i)}{6\alpha(\log (80\alpha) + \log \log (\tr{n}+1))} - \beta.
	$
\end{lemmarep}

\begin{appendixproof}
	Assume $|{X}_i| = k$. By \Cref{nblock}, we know that for every $1 \leq j \leq \lfloor k/\alpha \rfloor-1$ and every $\tr{\ite}_x \in \block_{i,j+1}$, we have $\widetilde{\valu}_i(\matching(\block_{i,j})) \geq \widetilde{\valu}_i(\{\tr{\ite}_x\})$. Averaging over all items in $\block_{i,j+1}$, we obtain:
	\begin{equation}\label[ineq]{blockav}
		\widetilde{\valu}_i(\matching(\block_{i,j})) \geq \frac{1}{\alpha} \sum_{\tr{\ite}_x \in \block_{i,j+1}} \widetilde{\valu}_i(\{\tr{\ite}_x\}).
	\end{equation}
	Also, recall that some items in ${X}_i$ may not have been grouped since $k$ may not be divisible by $\alpha$. For brevity, we group these items into block $\block_{i,\lfloor k/\alpha \rfloor+1}$. Unlike the other blocks, the size of this block is $k \textsf{ mod } \alpha$. However, since these items have the smallest values in $\widetilde{\valu}_i$, we have  
	\begin{equation}\label[ineq]{blockav2}
		\widetilde{\valu}_i(\matching(\block_{i,\lfloor k/\alpha \rfloor})) \geq \frac{1}{k \textsf{ mod } \alpha} \sum_{\tr{\ite}_x \in \block_{i,\lfloor k/\alpha \rfloor+1}} \widetilde{\valu}_i(\{\tr{\ite}_x\}) \geq \frac{1}{\alpha} \sum_{\tr{\ite}_x \in \block_{i,\lfloor k/\alpha \rfloor+1}} \widetilde{\valu}_i(\{\tr{\ite}_x\}).	
	\end{equation}
	
	Now, we have 
	\begin{align*}
		\widetilde{\valu}_i({A}^{\easily}_i)
		& = \sum_{1\leq j \leq \lfloor k/\alpha\rfloor} \widetilde{\valu}_i(\matching(\block_{i,j}))  &\mbox{$\widetilde{\valu}_i$ is additive}\\
		&\geq \sum_{1\leq j \leq \lfloor k/\alpha\rfloor-1} \frac{1}{\alpha} \sum_{\tr{\ite}_x \in \block_{i,j+1}} \widetilde{\valu}_i(\{\tr{\ite}_x\}) + \frac{1}{\alpha} \sum_{\tr{\ite}_x \in \block_{i,\lfloor k/\alpha \rfloor+1}} \widetilde{\valu}_i(\{\tr{\ite}_x\})  &\textit{\Cref{blockav,blockav2}}\\
		&= \frac{1}{\alpha}  \sum_{2\leq j \leq \lfloor k/\alpha\rfloor} \sum_{\tr{\ite}_x \in \block_{i,j}} \widetilde{\valu}_i(\{\tr{\ite}_x\}) \\
		&\geq \frac{1}{\alpha} \left(\sum_{\tr{\ite}_x \in \tr{X}_i} \widetilde{\valu}_i(\{\tr{\ite}_x\}) - \sum_{\tr{\ite}_x \in \block_{i,1}} \widetilde{\valu}_i(\{\tr{\ite}_x\})\right) \\ 
		&= \frac{1}{\alpha} \widetilde{\valu}_i({X}_i) - \frac{1}{\alpha}  \sum_{\tr{\ite}_x \in \block_{i,1}} \widetilde{\valu}_i(\{\tr{\ite}_x\})\\
		&\geq \frac{1}{\alpha} \widetilde{\valu}_i({X}_i) - \beta. \numberthis \label[ineq]{ineqf}
	\end{align*}
	The last line follows from the fact that the value of each item for each agent is bounded by $\beta$. On the other hand, we have
	\begin{align*}
		\tr{\valu}_i(A_i) 		&\leq 2\tr{\valu}_i({X}_i) &\tr{\agent}_i \in \easily\\
		& \leq 6\widetilde{\valu}_i({X}_i)(\log (80\alpha) + \log \log (\tr{n}+1))\numberthis\label[ineq]{ineqs} &\text{Inequality \eqref{delta}}.
	\end{align*}
	Combining \Cref{ineqf,ineqs} implies: 
	$$
	{\widetilde\valu}_i({A}^{\easily}_i) \geq \frac{\tr{\valu}_i(A_i)}{6\alpha(\log (80\alpha) + \log \log (\tr{n}+1))} - \beta.
	$$
\end{appendixproof}

\Cref{alg:allocate_easy} shows a pseudocode of the process described in this section. Note that, by the construction of bundles in \( \mathcal{A}^\easily \), each item is allocated to at most one agent in \( \easily \).

\begin{algorithm}[t]
	\caption{\textsc{ResolveEasyIntersections}}
	\LinesNumbered
	\label{alg:allocate_easy}
	
	\KwIn{Set of easily-satisfiable agents $\easily$} 
	\KwOut{Allocation $\mathcal{A}^{\easily}$ } 
	
	\ForEach { $\tr{\agent}_i \in \easily$}{ 
		Solve \Cref{eq:lp} to obtain $\widetilde{{\valu}}_i$\;
	}
	Construct $G^\easily$ and find a maximum matching $\matching$ in $G^\easily$\;
	\ForEach {$\tr{\agent}_i \in \easily$}{
		$A^{\easily}_i = \{ \matching(\block_{i,1}), \matching(\block_{i,2}), \dots, \matching(\block_{i,\lfloor k/\alpha \rfloor}) \}$  \;
	}
	\textbf{Return} $\mathcal{A}^{\easily}$.
	
\end{algorithm}

\subsection{Not-easily-satisfiable Agents}\label{sec:noteasy}
Recall that for each agent $\tr{\agent}_i \in \noteasily$, and for each $X_i \subseteq \alloc_i$ such that $\tr{\valu}_i(X_i) \geq \tr{\valu}_i(A_i)/2$, we have $|X_i| > 80\alpha(\log \tr{n} + 1)$. Now, we construct an allocation $\mathcal{A}^\noteasily$ based on multiallocation $\mathcal{A}$ for not-easily-satisfiable agents. We determine the bundles of these agents in $\mathcal{A}^\noteasily$ using a simple probabilistic process: for each item belonging to multiple agents of $\noteasily$ in $\mathcal{A}$, we randomly select one of those agents to have it in $\mathcal{A}^\noteasily$.

\Cref{alg:noteasily} provides the pseudocode for constructing the allocation $\mathcal{A}^\noteasily$. Let $\alloc^{\noteasily}_i$ denote the bundle assigned to agent $\tr{\agent}_i \in \noteasily$ in $\mathcal{A}^\noteasily$. By construction, each item in $\mathcal{A}^\noteasily$ is allocated to at most one agent in $\noteasily$. 
For reasons that will become clear in the next section, we defer proving any lower bound on $\tr{\valu}_i(\alloc^\noteasily_i)$ at this stage and instead provide a unified analysis for not-easily-satisfiable agents after the merging process.

\begin{algorithm}[t]
	\caption{\textsc{ResolveHardIntersections}}
	\label{alg:noteasily}
	\LinesNumbered
	\KwIn{Set of agents $\noteasily$, allocation $\mathcal{A}$}
	\KwOut{Allocation $\mathcal{A}^\noteasily$}
	Initialize $\mathcal{A}$ as empty allocation\;
	\ForEach{ $\tr{b}_x$}{
		Let $S_{\tr{b}_x} = \{ \tr{\agent}_i \mid \tr{b}_x \in \alloc_i, \tr{\agent}_i \in \noteasily \}$ \;
		Select \( \tr{\agent}_j \) uniformly at random from \( S_{\tr{b}_x} \)\;
		$\alloc^{\noteasily}_j \leftarrow \alloc^{\noteasily}_j \cup \{\tr{\ite}_x\}$ \;
	}
	\Return{$\mathcal{A}^\noteasily$}
\end{algorithm}

\subsection{Merging $\mathcal{\alloc}^\easily$ and $\mathcal{\alloc}^\noteasily$ }

Up to this point, the bundles of $ \mathcal{\alloc}^\easily $ and $ \mathcal{\alloc}^\noteasily $ are guaranteed to be disjoint within their respective categories. Specifically, for every $ \tr{\agent}_i \neq \tr{\agent}_j \in \easily $, we have $ {\alloc}^{\easily}_i \cap \alloc^{\easily}_j = \emptyset $. Similarly, for every $ \tr{\agent}_i \neq \tr{\agent}_j \in \noteasily $, we have $ {\alloc}^{\noteasily}_i \cap \alloc^{\noteasily}_j = \emptyset$. However, an item may belong to one agent in $ \easily $ and another in $ \noteasily $. We refer to such an item as shared and to any item belonging to only one agent in $\mathcal{\alloc}^\easily \cup \mathcal{\alloc}^\noteasily$ as unique.

We construct allocation $\mathcal{\alloc}'$ as follows: each unique item in $\mathcal{\alloc}^\easily \cup \mathcal{\alloc}^\noteasily$ is allocated to the same agent in $\mathcal{\alloc}'$. In other words, if an item $\tr{\ite}_j$ is allocated only to agent $\tr{\agent}_i$ in $\mathcal{\alloc}^\easily \cup \mathcal{\alloc}^\noteasily$, it is also allocated to $\tr{\agent}_i$ in $\mathcal{\alloc}'$.

The challenge is to resolve the issue of shared items. For this, we partition the bundles of the agents of $\easily$ in $\mathcal{\alloc}^\easily$  into blocks of size 2. Let $\tr{\agent}_i$ be an agent in $ \easily$. We sort the items in \( \mathcal{\alloc}^\easily \) in descending order of their values under the auxiliary valuation function \( \widetilde{\valu}_i \). Then, we group them into blocks, where each block consists of two consecutive items from this sorted order. Specifically, the two most valuable items form the first block, the next two form the second block, and so on, resulting in blocks \( \block_{i,1}, \block_{i,2}, \dots, \block_{i,k/2} \).
To make sure that all the items are grouped, if $k$ is not divisible by 2, we introduce a dummy item   with  value $0$ for all agents to $\mathcal{\alloc}^\easily$.

Now, we have three types of blocks: (1) blocks where both items are unique, (2) blocks where one item is unique and the other one is shared, and (3) blocks where both items are shared. The first type requires no action since both items are already allocated in $\mathcal{\alloc}'$. For the second type, we allocate the shared item to the not-easily-satisfiable agent who owns it in $\mathcal{A}^\noteasily$.

The only remaining case is for items in blocks of the third type. Let $\block_{i,j}$ be such a block, containing items $\tr{\ite}_y$ and $\tr{\ite}_z$. We handle this by considering two cases:

\paragraph{Case 1. $\tr{\ite}_y$ and $\tr{\ite}_z$ belong to distinct agents in $\mathcal{\alloc}^\noteasily$:}  
In this case, we randomly select either $\tr{\ite}_y$ or $\tr{\ite}_z$ and allocate it to $\tr{\agent}_i$ in $\mathcal{\alloc}'$. The remaining shared item in the block is then allocated to the not-easily-satisfiable agent who owns it in $\mathcal{\alloc}^\noteasily$.

\paragraph{Case 2. $\tr{\ite}_y$ and $\tr{\ite}_z$ belong to the same agent in $\mathcal{\alloc}^\noteasily$:}  We resolve these conflicts by looking at the bundles of agents in $\noteasily$. Suppose all other conflicts have been resolved, except for the current case. For each agent $\tr{\agent}_i \in \noteasily$, we examine the blocks where both items are shared by $\tr{\agent}_i$. We ask agent \( \tr{\agent}_i \) to select one item from each block to maximize her total value when combined with her existing bundle in \( \alloc'_i \).
We then allocate these selected items to $\alloc'_i$ and allocate the remaining items to the easily-satisfiable agents who own them in $\mathcal{\alloc}^\easily$. 

Once the above process ends, we are guaranteed that all bundles of $\mathcal{A}'$ are disjoint. \Cref{alg:construct_allocation} shows a pseudocode of our algorithm for merging $\mathcal{\alloc}^\easily$ and  $\mathcal{\alloc}^\noteasily$.

In \Cref{lem:31,lem:32} we prove that with a non-zero probability, the final allocation has the desired approximation guarantee. 

\begin{lemmarep}\label{lem:31}
	For every agent $\tr{\agent}_i \in {\easily}$, we have $\tr{\valu}_i(\alloc'_i) \geq \widetilde{{\valu}}_i(\alloc^\easily_i)/2-\beta$.
\end{lemmarep}	
\begin{appendixproof}
	The idea to prove this lemma is similar to \Cref{lem:easily}. Recall that in the merging step, $\alloc^\easily_i$ is partitioned into blocks of size 2, and at least one item from each block appears in $\alloc'_i$. Let the blocks of $\alloc^\easily_i$ be denoted as $\block_{i,j} = \{\tr{\ite}_{x_j}, \tr{\ite}_{x'_j}\}$ for every $1 \leq j \leq |\alloc^\easily_i|/2$. 
	Since the blocks are constructed in decreasing order of their auxiliary values for agent $\tr{\agent}_i$, we have:  
	\begin{equation}\label[ineq]{blockshift}
		\min\big(\widetilde{\valu}_i(\{\tr{\ite}_{x_j}\}), \widetilde{\valu}_i(\{\tr{\ite}_{x'_j}\})\big) \geq \max\big(\widetilde{\valu}_i(\{\tr{\ite}_{x_{j+1}}\}), \widetilde{\valu}_i(\{\tr{\ite}_{x'_{j+1}}\})\big).
	\end{equation}
	Now, we have
	\begin{align*}
		\tr{\valu}_i(\alloc'_i) 		&\geq \widetilde{\valu}_i(\alloc'_i) \\
		&\geq \sum_{\tr{\ite}_y \in \alloc'_i} \widetilde{\valu}_i(\{\tr{\ite}_y\}) & \mbox{$\widetilde{\valu}_i$ is additive}\\
		&\geq \sum_{1 \leq j \leq |\alloc^\easily_i|/2} \min(\widetilde{\valu}_i(\{\tr{\ite}_{x_j}\}), \widetilde{\valu}_i(\{\tr{\ite}_{x'_j}\}))\\
		&\geq \sum_{2 \leq j \leq |\alloc^\easily_i|/2} \max(\widetilde{\valu}_i(\{\tr{\ite}_{x_j}\}),\widetilde{\valu}_i(\{\tr{\ite}_{x'_j}\})) &\text{\Cref{blockshift}}\\
		&\geq  \frac{1}{2}\sum_{2 \leq j \leq |\alloc_i^\easily|/2} \widetilde{\valu}_i(\{\tr{\ite}_{x_j}\})+\widetilde{\valu}_i(\{\tr{\ite}_{x'_j}\})\\
		& = \frac{1}{2} \widetilde{\valu}_i(\alloc^\easily_i) - \frac{1}{2} (\widetilde{\valu}_i(\{\tr{\ite}_{x_1}\}) +\widetilde{\valu}_i(\{\tr{\ite}_{x'_1}\}))\\
		& \geq \frac{1}{2} \widetilde{\valu}_i(\alloc^\easily_i)- \beta.
	\end{align*}

\end{appendixproof}

\begin{lemmarep}\label{lem:32} With the probability more than $1-1/n$, for every agent \( \tr{\agent}_i \in \noteasily \),  
	$	\tr{\valu}_i(\alloc'_i) \geq {\tr{\valu}_i(\alloc_i)}/{480}.
	$
\end{lemmarep}
\begin{appendixproof}
	Fix an agent $\tr{\agent}_i \in \noteasily$. We say two shared items are \textit{linked} if they belong to $\tr{\agent}_i$ and are also shared with the same agent in $\easily$ under allocation \( \alloc^\easily \). Let \( L \) be the set of such linked items in \( \alloc^\noteasily_i \).  
	To simplify the proof of \Cref{lem:32}, we first provide an intuitive interpretation of the merging step for \( \tr{\agent}_i \), which can be thought of as following steps:  
	\renewcommand{\labelenumi}{\arabic{enumi}.}
	\begin{enumerate}
		\item Each unique item in \( \alloc^\noteasily_i \) is added to \( \alloc'_i \) with probability 1.  
		\item Each shared item in \( \alloc^\noteasily_i \setminus L \) is added to \( \alloc'_i \) with probability \( 1/2 \).  
		\item All items in \( L \) are added to \( \alloc'_i \) with probability 1.  
		\item Finally, \(\tr{\agent}_i\) selects a feasible subset of items from \( L \) to remove from her current bundle \(\alloc'_i\). Specifically, she removes one item from each pair of linked items in \( L \) and reallocates it to the easily-satisfiable agent who owns it in \(\mathcal{\alloc}^\easily\). She makes this selection to maximize the value of her remaining bundle.
	\end{enumerate}  
	
	Let us now denote the bundle of agent \( \tr{\agent}_i \) in \( \mathcal{\alloc'} \) just before step 4 as \( \alloc''_i \). Based on the first three steps, we can conclude that each item in \( \alloc^\noteasily_i \) appears in \( \alloc''_i \) with probability at least \( 1/2 \). Moreover, since each item in \( \alloc_i \) appears in \( \alloc^\noteasily_i \) with probability at least \( 1/\alpha \), we can say that each item in \( \alloc_i \) appears in \( \alloc''_i \) with probability at least \( 1/(2\alpha) \). Additionally, since \( \tr{\agent}_i \) is a not-easily-satisfiable agent, by definition, for every subset \( X_i \subseteq \alloc_i \) with \( |X_i| \leq 80\alpha(\log \tr{n} + 1) \), we have \( \tr{\valu}_i(X_i) < \tr{\valu}_i(\alloc_i)/2 \). Applying \Cref{lemma:1} implies that with probability more than $1-1/n$ we have 
	$	\tr{\valu}_i(\alloc''_i) \geq {\tr{\valu}_i(\alloc_i)}/({240\alpha}).
	$
	Now, consider Step 4 and let $\bar L$ denote the set of items removed by agent $\tr{\agent}_{i}$, i.e., $\alloc'_i = \alloc''_i \setminus \bar L$. Since $\tr{\agent}_{i}$ removes items to maximize her value for the remaining items, $\bar L$ is the optimal feasible subset of items to remove compared to any other feasible subset, including $L \setminus \bar L$. This  means:
	\begin{equation}\label[ineq]{ineq:final}
		\tr{\valu}_i\left(\alloc'_i\right) = \tr{\valu}_i\left(\alloc''_i \setminus \bar L\right) \geq \tr{\valu}_i\left(\alloc''_i \setminus \left(L \setminus \bar L\right)\right).
	\end{equation}
	By the subadditivity assumption, this means \begin{equation}\label[ineq]{ineq:final2}
		\tr{\valu}_i(\alloc'_i) \geq \tr{\valu}_i(\alloc_i'')/2.\end{equation}
	Combining \Cref{ineq:final,ineq:final2} implies that with probability at least $1-1/n$ we have
	$
	\tr{\valu}_i(\alloc'_i) \geq {\tr{\valu}_i(\alloc_i)}/({480\alpha}).
	$
\end{appendixproof}

Finally, we can prove \Cref{lemma:second} by combining \Cref{lem:easily,,lem:31,,lem:32}.

\begin{proof}[ of \Cref{lemma:second}] Consider allocation $\mathcal{A}'$. 
	For every $\tr{\agent}_i \in \easily$, by  \Cref{lem:easily}  we have $
	\widetilde{{\valu}}_i(\alloc^\easily_i) \geq \frac{\tr{\valu}_i(\alloc_i)}{6\alpha(\log (80\alpha) + \log \log (\tr{n}+1))} - \beta,
	$ and by  \Cref{lem:31} we have      $\tr{\valu}_i(\alloc'_i) \geq \widetilde{{\valu}}_i(\alloc^\easily_i)/2-\beta$. Hence,
	\begin{align}
		\tr{\valu}_i(\alloc'_i) \geq \frac{\tr{\valu}_i(\alloc_i)}{12\alpha(\log (80\alpha) + \log \log (\tr{n}+1))}- \frac{3\beta}{2}. \label[ineq]{easilybound}
	\end{align}
	Furthermore, for every agent $\tr{\agent}_i \in \noteasily$, by \Cref{lem:32}, with a  probability more than $1-1/n$ we have  
	$
	\tr{\valu}_i(\alloc'_i) \geq \tr{\valu}_i(\alloc_i)/(480\alpha).
	$
	By the union bound, This inequality holds for all not-easily satisfiable agents with a non-zero probability. This implies that there exists an outcome where for every agent \( \tr{\agent}_i \) we have 
	\begin{align*}
		\tr{\valu}_i(\alloc'_i) &\geq \tr{\valu}_i(\alloc_i) \cdot \min\left( \frac{1}{480\alpha},\frac{1}{12\alpha(\log (80\alpha) + \log \log (\tr{n}+1))}- \frac{3\beta}{2} \right)\\
		&\geq \frac{\tr{\valu}_i(\alloc_i)}{480\alpha(\log (80\alpha) + \log \log (\tr{n}+1))}- \frac{3\beta}{2}.
	\end{align*}
\end{proof}

\begin{algorithm}[t]
	\caption{Merging $\mathcal{\alloc}^\easily$ and $\mathcal{\alloc}^\noteasily$}
	\LinesNumbered
	\label{alg:construct_allocation}
	
	\KwIn {Allocations \( \mathcal{A}^\easily \), \( \mathcal{A}^\noteasily \)}
	\KwOut{ Allocation \( \mathcal{A}' \)}
	\For{each unique item $\tr{\ite}_j \in \mathcal{\alloc}^\easily \cup \mathcal{\alloc}^\noteasily$}
	{ Allocate $\tr{\ite}_j$ to the same agent in $\mathcal{\alloc}'$ }

	\For{each agent $\tr{\agent}_i \in \easily$}
	{ 
		Sort items in $\mathcal{\alloc}^\easily_i$ by $\widetilde{\valu}_i$ in descending order and 
		partition into blocks $\block_{i,1}, \block_{i,2}, \ldots$ of size 2\;
	}

	\For{each block $\block_{i,j}$}
	{
		\If{both items are unique}
		{ \textbf{Continue} \tcp*{Items have already resolved}}
		\ElseIf{one item is unique, one is shared}
		{ Allocate the shared item to its owner in $\mathcal{\alloc}^\noteasily$\; }
		\ElseIf{both items are shared and belong to different agents in $\mathcal{\alloc}^\noteasily$}
		{ 
			Randomly allocate one of them to $\tr{\agent}_i$ in $\mathcal{\alloc}'$\;
			Allocate the remaining item to its owner in $\mathcal{\alloc}^\noteasily$\;
		}
	}
	
	\For{each agent $\tr{\agent}_i \in \noteasily$}
	{ 
		$L =$ Set of \textit{linked} items belonging to $\tr{\agent}_i$\;
		\For{each pair of linked items in $L$}
		{
			\textbf{Ask} $\tr{\agent}_i$ to choose one item to keep \tcp*{$\tr{\agent}_i$ optimizes her bundle in $\mathcal{\alloc}'_i$}
			\textbf{Allocate} the selected item to $\tr{\agent}_i$ in $\mathcal{\alloc}'$\;
			\textbf{Allocate} the other item to the easily satisfiable agent in $\mathcal{\alloc}^\easily$\;
		}
		
	}

\end{algorithm}


\section{Warm-up, A ${1}/{O(\log n \log \log n)}$-$\MMS$ Allocation Algorithm}\label{sec:warmup}

For agents with subadditive valuations, \textcite{maseed123} prove the existence of a $\nicefrac{1}{O(\log n \log \log n)}$-$\MMS$ allocation. 
As a warmup, here we show that \Cref{lemma:second} offers a much simpler proof for the existence of such allocations.

Before presenting the algorithm, we start with a simple yet important lemma about the maximin-share criteria. This lemma captures a fundamental property of $\MMS$ and has been widely used in previous studies.

\begin{lemma}[\cite{ghodsi2018fair,amanatidis2015approximation} - Restated]\label{reduction}
	If a \(\beta\)-$\MMS$ allocation exists for instances where no single item's value for any agent exceeds \(\beta\), then a \(\beta\)-$\MMS$ allocation exists even without this assumption.
\end{lemma}

\Cref{reduction,lemma:second} directly imply \Cref{lemma:reduction}.

\begin{lemmarep}\label{lemma:reduction}
	The existence of an $\alpha$-multiallocation that guarantees a $\frac{1}{\eta}$-\MMS\ approximation for the subadditive maximin share problem leads to the existence of a 	
	\begin{equation}\label{approxl}
		\left(\frac{1}{10800\alpha\eta(\log \alpha + \log \log n)}\right)\textsf{-MMS}
	\end{equation}
	guarantee for the subadditive maximin share problem.
\end{lemmarep}

\begin{appendixproof}
	We begin by noting that for \( \alpha = 1 \), a \( 1 \)-multiallocation is a valid allocation, and in this case, the statement of \Cref{lemma:reduction} trivially holds.
	For \( \alpha \geq 2 \), we use \Cref{reduction}. Let 
	\begin{equation}\label{eq:beta}
		\beta = \frac{1}{10800 \alpha\eta(\log \alpha + \log \log n)}.
	\end{equation}   Our goal is to prove that a $\beta\text{-\MMS} $ allocation always exists. By \Cref{reduction}, without loss of generality, we assume that the value of any item for any agent is less than $ \beta$.
	
	Given the existence of an \( \alpha \)-multiallocation, \Cref{lemma:second} guarantees that there exists a $$ \left(\frac{1}{480 \alpha \eta (\log (80\alpha) + \log (\log n + 1))} - \frac{3\beta}{2}\right)\text{-\MMS}$$ allocation.
	For \( \alpha \geq 2 \), we know that \( \log(80\alpha) \leq 8 \log \alpha \). Hence, we have:
	
	\begin{align*}
		480 \alpha \eta \left( \log(80\alpha) + \log(\log n + 1) \right)
		&\leq 480 \alpha \eta \left( 8 \log \alpha + \log(\log n + 1) \right)\\
		&\leq 480 \alpha \eta \left( 8 \log \alpha + \log(2 \log n) \right) & n \geq 2 \\
		&\leq 480 \alpha \eta \left( 8 \log \alpha + \log 2 + \log (\log n) \right)\\
		&
		\leq 480 \alpha \eta \left( 9 \log \alpha + \log (\log n) \right) \quad &\alpha \geq 2\\
		&
		\leq 4320 \alpha \eta(\log \alpha + \log \log n)\\
		&=\frac{2}{5\beta}. &\Cref{eq:beta}.
	\end{align*}
	
	Thus, we obtain a $$\left(\frac{5\beta}{2} - \frac{3\beta}{2}\right)\text{-\MMS } = \beta\text{-\MMS}$$  allocation, which is the desired guarantee.
\end{appendixproof}

The algorithm presented by \textcite{maseed123} is based on the following lemma.

\begin{lemma}[\cite{maseed123}, Lemma 6.1] \label{seddighin}
	For any instance of the fair allocation problem with subadditive valuations, there always exists an allocation that guarantees $1/4$-$\MMS$ to at least $ n/3$\footnote{In \cite{maseed123}, the lemma is stated using \( \lfloor n/3 \rfloor \) instead of \( n/3 \). However, their proof is valid for \( n/3 \).} of the agents.
\end{lemma}

Building on \Cref{seddighin}, \textcite{maseed123} suggest the following algorithm:  construct a multiallocation by iteratively selecting a subset of agents and allocating them bundles. At each step $t$, assuming $n_t$ agents have remained, choose a subset of size at least  $ n_t/3$ of the agents and allocate each agent a bundle of value at least \( 1/4 \) using \Cref{seddighin}. These agents are then removed, and the process continues with the remaining agents and all the items. 

Since at each step, a fraction $1/3$ of the remaining agents receive a bundle, we have
$
n_{t+1} \leq {2n_{t}}/{3}, 
$
which implies that the process completes in at most $ \log_{3/2} n + 2$ steps. Therefore, at the end of the algorithm, each item is allocated at most $ \log_{3/2} n  + 2$ times. Denote this multiallocation by $\mathcal{A}$.

At this stage, without \Cref{lemma:second}, \textcite{maseed123} further complicate the process by modifying the agents' valuation functions to ensure that each agent receives at least \( \log n \) bundles rather than one. This adjustment slightly weakens the approximation guarantee but ensures that at least one of the \(\log n\) bundles retains significant value during the conversion from multiallocation to allocation. Here, we can directly apply \Cref{lemma:reduction} to \( \mathcal{A} \), which proves the existence of a $\nicefrac{1}{O(\log n \log \log n)}$-\(\MMS \) allocation.

\begin{theoremrep}\label{theorem:w1}
	The maximin share problem with subadditive agents admits a
	$\frac{1}{648000 \log n  \log \log n}\text{-}\MMS
	$ guarantee.
\end{theoremrep}
\begin{appendixproof}
	For $n=2$, a simple cut and choose algorithm guarantees $1$-$\MMS$ to both the agents. For $n\geq3$, we compute the approximation guarantee of the algorithm. 
	The algorithm terminates after at most $\log_{3/2} n + 2$ steps. For $n\geq 3$, we have $\log_{3/2} n + 2 \leq 3\log n$. Thus, the algorithm terminates in at most $3\log n$ steps and produces a $3\log n$-multiallocation where each agent values their bundle at least $1/4$.
	Setting $\alpha = 3\log n$ and $\eta = 4$ in \Cref{approxl} implies that a
	\[
	\frac{1}{10800 \cdot 3\log n \cdot 4 (\log(3\log n) + \log \log n)}\text{-}\MMS
	\]
	allocation exists. The logarithmic term, \(\log(3\log n) + \log \log n\), simplifies to  
	$
	\log 3 + 2\log \log n. 
	$  For $n\geq 3$, we have $ \log 3 + 2\log \log n \leq 5 \log \log n$.
	Therefore, the denominator becomes  
	$
	129600 \log n (5\log \log n) = 648000 \log n \log \log n.
	$  
	Thus, the algorithm guarantees the existence of a  
	\[
	\frac{1}{648000 \log n  \log \log n}\text{-}\MMS
	\] 
	allocation.
\end{appendixproof}

\section{Sublogarithmic Approximation Guarantee for Polynomially Many Items }\label{sec:sublogarithmicapprox}
In this section, we propose an algorithm that achieves an improved $\MMS$ guarantee when the number of items is $\textsf{poly}(n)$.  

A key component in breaking the logarithmic barrier is \Cref{lemma:third}.  

\begin{lemmarep}\label{lemma:third}
	Let $Q \subseteq \agents$ be a subset of agents. For $k = |\agents| / (6|Q|)$, there exists a subset $Q' \subseteq Q$ and $\lceil  k\rceil $ disjoint allocations  $\mathcal{\alloc}^1, \mathcal{\alloc}^2,\ldots, \mathcal{\alloc}^{\lceil  k\rceil}$ of items to agents of $Q'$ such that  $|Q'| \geq |Q|/6$ and for every agent $a_{i} \in Q'$ and $1 \leq j \leq \lceil  k\rceil$ we have 
	$\valu_{i}(A^j_{i}) \geq 1/4.$ 
\end{lemmarep}

\begin{appendixproof}
	Without loss of generality, we assume  $|Q| < |N|/6$; otherwise, we have $ \lceil k \rceil=  1$,  in which case \Cref{lemma:third}  directly follows from  \Cref{seddighin}.

	Define $Q^*$ as the set containing $6\lfloor k\rfloor$ copies of each agent in $Q$. That is, for every $\agent_i \in Q$, there exist $6\lfloor k\rfloor$ agents with the same valuation function as $\agent_i$ in $Q^*$, denoted by  $ \agent_{i,1}, \agent_{i,2}, \dots, \agent_{i,6\lfloor k\rfloor}. $
	
	Consider an instance with the agents in $Q^*$. Since $|Q^*| = 6\lfloor k\rfloor|Q| \leq |\agents|$, and the $\MMS$ value of the agents in an instance with $|N|$ agents is  $1$, each agent in this instance has an $\MMS$ value of at least 1.
	By \Cref{seddighin}, there exists an allocation for the agents in $Q^*$ that guarantees a value of at least $1/4$ for at least $\lfloor |Q^*|/3 \rfloor = 2\lfloor k\rfloor|Q|$ agents in $Q^*$. Let \(\mathcal{A}\) be this allocation. For each agent \(\agent_i \in Q\), define \(C_i\) as the set of copies of \(\agent_i\) that receive a value of at least \(1/4\) in this allocation, i.e.,
	$
	C_i = \{ \agent_{i,j} \mid \valu_i(\alloc_{i,j}) \geq 1/4 \},
	$
	where \(\alloc_{i,j}\) denotes the bundle allocated to \(\agent_{i,j}\) in \(\mathcal{A}\).
	Thus, the total number of agents in $Q^*$ who receive a value of at least $1/4$ in $\mathcal{\alloc}$ is $\sum_{\agent_i \in Q} |C_{i}|$. We now show that for at least $|Q|/6$ of the sets $C_{i}$, we have $|C_{i}| \geq \lceil k \rceil$.  
	Assume, for contradiction, that this is not the case. Then, we have:  
	\begin{align*}
		\sum_{1 \leq i \leq |Q|} |C_i|  
		&< \frac{|Q|}{6} \cdot 6\lfloor k \rfloor + \left(|Q| - \frac{|Q|}{6}\right) \lfloor k \rfloor\\ 
		&= |Q| \lfloor k \rfloor  + \frac{5|Q| \lfloor k\rfloor}{6}\\  
		&= \frac{11|Q| \lfloor k \rfloor}{6} \\ 
		&<  2|Q| \lfloor k\rfloor.
	\end{align*}
	However, this contradicts the fact that at least $2\lfloor k\rfloor |Q|$ agents receive bundles with a value of at least $1/4$ in $\mathcal{\alloc}$. Hence, for at least $\lceil |Q|/6\rceil$ of the sets $C_i$, it holds that $|C_i| \geq k$.
	
	Finally, define $Q' = \{\agent_{{i}} \mid \agent_{{i}}\in Q \text{ and } |C_{i}| \geq k \}$ and consider an arbitrary ordering of the agents in each $C_{i}$. For every $1 \leq j \leq \lceil k \rceil$, define  allocation $\mathcal{\alloc}^j$ as follows:
	$$
	\forall \agent_{i} \in Q' \qquad \mathcal{\alloc}^j_{i} = \text{the bundle allocated to the } j\text{'th agent in } C_{i}. 
	$$
	These allocations satisfy the conditions of \Cref{lemma:third}.
	
\end{appendixproof}

We now leverage \Cref{lemma:third} to design an algorithm that provides a better $\MMS$ approximation guarantee for instances with subadditive valuations when $m$ is $\textsf{poly}(n)$.  \Cref{algorithm:poly} outlines our method for this case. The algorithm is simple: We start with the entire set of agents. At each step $ t $, let \( \agents_t \) be the set of remaining agents.  By \Cref{lemma:third}, we can select at least \( \lceil |\agents_t|/6\rceil  \) of the agents in \( \agents_t \) and allocate bundles to each of them, such that each selected agent receives $\lceil |\agents|/(6|\agents_t|) \rceil$ bundles, each worth at least $1/4$ to her. We then remove these agents and repeat the process with the remaining agents and all items until no agent remains. The idea is that as we remove more agents, we can allocate additional bundles to the remaining agents using \Cref{lemma:third}. Finally, for each agent who receives multiple bundles, we randomly select one of their bundles with equal probability as their allocated bundle.  We then apply \Cref{alg:lem2} to convert this multiallocation into an allocation.

\begin{algorithm}[t]
	\caption{Sub-Polylogarithmic Approximation Algorithm}
	\label{algorithm:poly}
	\begin{algorithmic}[1]
		\STATE \textbf{Input:} A set of agents $\agents$ and a set of items $\items$
		\STATE \textbf{Output:} An allocation of items to agents
		\STATE Initialize $\agents_1 \gets \agents$ \tcp*{set of remaining agents}
		\STATE Initialize $\mathcal{A}$ as an empty allocation\;
		\STATE $i \gets 1$\;
		\WHILE{$\agents_i \neq \emptyset$}
		\STATE Select a set $S$ of at least $\lceil  |\agents_i|/6 \rceil$ agents and allocate $\lceil \frac{|\agents|}{6|\agents_i|} \rceil$\\ bundles worth at least $1/4$ to each in $S$
		\tcp*{Using \Cref{lemma:third}} 
		\STATE $\agents_{i+1} = \agents_{i} \setminus S$\;
		\ENDWHILE
		\FOR{each agent who received multiple bundles}
		\STATE Randomly select one of their bundles as the final multiallocation\;
		\ENDFOR
		\IF{an item is allocated more than $18\sqrt{\log_{5/6} m}$ times in $\mathcal{A}$}
		\STATE Goto Line 3\;	
		\ENDIF
		\STATE Convert $\mathcal{A}$ to an allocation \tcp*{Using \Cref{alg:lem2}}
	\end{algorithmic}
\end{algorithm}

Denote by $\mathcal{\alloc}$ the final multiallocation. In \Cref{lem:numberofoccurances}, we provide an upper bound on the number of bundles that contain each item in $\mathcal{\alloc}$. 
\begin{lemmarep}\label{lem:numberofoccurances}
	With a non-zero probability, $\mathcal{A}$ is a $(18\sqrt{\log_{5/6} m})$-multiallocation.
\end{lemmarep}
\begin{appendixproof}
	First, note that by \Cref{lemma:third}, at each step $t$, at least $ |\agents_t| / 6 $ of the remaining agents receive bundles. Therefore, we have
	\begin{align*}
		|N_{t+1}| &\leq |N_t| - |N_t|/{6}\\
		&\leq |N_t| \frac{5}{6}. \numberthis \label[ineq]{ineq:steps}
	\end{align*}
	\Cref{ineq:steps} implies that after $ t $ steps, we have
	\begin{equation}\label[ineq]{ineq:totalsteps}
		|\agents_{t+1}|\leq  (|\agents|) \left( \frac{5}{6} \right)^t.
	\end{equation} 
	Furthermore, the number of bundles each agent receives at step $t$ is at least	\begin{equation} \label[ineq]{number_of_bundles_at_each_step}
		\left\lceil  \frac{|\agents|}{ 6 |\agents_{t}|} \right\rceil\geq	\left\lceil \frac{|\agents|}{ 6 \left(|\agents| \left( \frac{5}{6} \right)^{t-1} \right)}  \right\rceil\geq \frac{1}{6}{\left(\frac{6}{5}\right)^{t-1}}  
	\end{equation}
	bundles. Denote by \( t^* \) the number of steps the algorithm takes before returning the multiallocation. By \Cref{ineq:totalsteps}, we obtain  
	$$
	t^* \leq \log_{6/5} |\agents| + 6.
	$$

	Given that the number of bundles each selected agent receives in step \( t\) is lower-bounded by \Cref{number_of_bundles_at_each_step}, and considering the method used to determine the final multiallocation, we can conclude that the probability of each item being allocated at step \( t \) is at most
	\begin{equation} \label[ineq]{ineq:itemprob}
		\min\left(\frac{6}{(\frac{6}{5})^{t-1}},1\right).
	\end{equation}
	Let $c=\sqrt{\log_{6/5} m}$. We now demonstrate that, with a non-zero probability, all items appear in fewer than \( 18c \) bundles. 
	
	Fix an item \( \ite_j \) and steps \( t_1, t_2, \ldots, t_{c} \) such that for all $1 \leq i \leq c$, $17c+1 \leq t_i \leq t^*$. Note that we can assume without loss of generality that $t^* \geq 18c$; otherwise, each item is allocated at most $t^* \leq 18c$ times, which directly implies \Cref{lem:numberofoccurances}. The probability that \( \ite_j \) appears in all these steps is
	\begin{align*}
		\Pr[\ite_j \text{ appears in all steps $t_1,t_2,\ldots,t_c$}] &\leq \prod_{i=1}^{c} \min\left( \frac{6}{\left( \frac{6}{5} \right)^{t_i - 1} }, 1 \right) &\text{\Cref{ineq:itemprob}}\\
		&\leq \prod_{i=1}^{c} \frac{6}{\left( \frac{6}{5} \right)^{t_i - 1}} \\
		&\leq \left( \frac{6}{\left( \frac{6}{5} \right)^{17c } } \right)^{c} & t_i\geq 17c+1\\
		&=  \left(6^{(1/c) }\cdot(\frac{5}{6})^{17}\right)^{c^2}.
	\end{align*}
	
	On the other hand, the number of ways to choose $c$ steps from those occurring after step $17c$ is at most  
	\begin{align*}
		\binom{t^*-17c}{c} &\leq \binom{\log_{6/5}|N|+6-17c}{c}\\
		&\leq \binom{\log_{6/5}|N|}{c} &c\geq 1\\
		&= \frac{(\log_{6/5}|N|)!}{c! (\log_{6/5}|N|-c)!}\\
		&\leq \frac{(\log_{6/5}|N|)^c}{c!}\\
		&\leq \frac{(\log_{6/5}|N|)^c}{(c/e)^c} &c! \geq (c/e)^c \\
		&\leq \left(\frac{e\log_{6/5} |\agents|}{c}\right)^{c} \\
		&\leq (ec)^{c} \\
		&= \left((ec)^{1/c}\right)^{c ^2},
	\end{align*}
	where $e \simeq 2.718$ is the base of natural logarithms. Therefore, the probability that an item $\ite_j$ appears in more than $c$ bundles is bounded by:
	$$
	\left(6^{(1/c) }\cdot(\frac{5}{6})^{17}\right)^{c^2} \cdot \left((ec)^{1/c}\right)^{c ^2} = \left((6ec)^{(1/c) }\cdot(\frac{5}{6})^{17}\right)^{c^2} < (\frac{5}{6})^{c^2} <\frac{1}{m}
	$$  
	By the union bound, this property holds for all items with a non-zero probability. Consequently, there exists an event such that each item is allocated to at most $c = \sqrt{\log_{6/5} m}$ agents who have received bundles in steps after $17\sqrt{\log_{6/5} m}$. Since each item is allocated at most once per step, it follows that during the first $17\sqrt{\log_{6/5} m}$ steps, each item is allocated at most $17\sqrt{\log_{6/5} m}$ times. Thus, in this event, every item is allocated in at most $18\sqrt{\log_{6/5} m}$ steps. 
\end{appendixproof}

Finally, \Cref{lem:numberofoccurances} combined with Lemma \ref{lemma:second}, yields Theorem \ref{theorem:sqlm}.

\begin{theoremrep}\label{theorem:w2}
	\label{theorem:sqlm}
	The maximin share problem with subadditive agents admits a $\frac{1}{12441600 \cdot \sqrt{\log m} \cdot \log \log m}$-$\MMS
	$ guarantee.
\end{theoremrep}
\begin{appendixproof}
	For $n=2$, a simple cut-and-choose algorithm guarantees $1$-$\MMS$ to both agents. Additionally, if $m \leq n$, we can trivially guarantee $1$-$\MMS$. For $n \geq 3$ and $m > n$, we compute the approximation guarantee of the algorithm.
	The algorithm terminates after at most $$18\sqrt{\log_{6/5} m} \leq 36\sqrt{\log m}$$ steps. This produces a $36\sqrt{\log m}$-multiallocation where each agent values their bundle at least $1/4$.
	
	Setting $\alpha = 36\sqrt{\log m}$ and $\eta = 4$ in \Cref{approxl}, the approximation ratio is:
	\[
	\frac{1}{10800 \cdot 36\sqrt{\log m} \cdot 4 (\log(36\sqrt{\log m}) + \log \log n)}\text{-}\MMS.
	\]
	
	For the logarithmic term, we simplify:
	\[
	\log(36\sqrt{\log m}) + \log \log n \leq \log 36 + \frac{1}{2}\log \log m + \log \log n.
	\]
	Since $\log 36 \leq 6$ and $\log \log n \leq \log \log m$ (because $m \geq n$), we have:
	\[
	\log 36 + \frac{1}{2}\log \log m + \log \log m \leq 6 + 2\log \log m\leq 8\log\log m
	\]
	Thus, the denominator becomes:
	\[
	10800 \cdot 36 \cdot 4 \cdot 8 \cdot \sqrt{\log m} \cdot \log \log m = 12441600 \cdot \sqrt{\log m} \cdot \log \log m.
	\]
	
	The algorithm guarantees the existence of a
	\[
	{\frac{1}{12441600 \cdot \sqrt{\log m} \cdot \log \log m}\text{-}\MMS}.
	\]
	allocation.
\end{appendixproof}

\section{Main Contribution: A $1/O((\log \log n)^2)$-\MMS\ Guarantee}\label{sec:main}

The main contribution of this section is Lemma~\ref{lemma:fourth} which we prove first and then use it to state our main theorem (Theorem~\ref{theorem:main}). 

\begin{lemma} \label{lemma:fourth}
	Let $Q \subseteq N$ be a subset of agents. There exists a subset $Q' = \{\agent_{x_1}, \agent_{x_2}, \ldots,\agent_{x_{|Q'|}}\} \subseteq Q$ and an allocation  $A_{x_1},A_{x_2},\ldots,A_{x_{|Q'|}}$ of items to agents of $Q'$ such that $|Q'| \geq |Q|\frac{k}{k+1}$ for $k = \lfloor n/|Q| \rfloor$ and for every agent $\agent_{x_i} \in Q'$  we have 
	$\valu_{x_i}(A_{x_i}) \geq 1/2.$
\end{lemma}

Recall that in Lemma~\ref{lemma:fourth} we have an \MMS\ problem in which every agent has $n$ disjoint choices, each of which provides a value of at least 1 to her. However, we only wish to allocate items to a smaller subset of agents $Q$. Recall that we denote the ratio of $n$ over $|Q|$ by $k = \lfloor n/|Q| \rfloor$. The goal is to prove that there exists an allocation that provides a value of at least $1/2$ to at least a $\frac{k}{k+1}$ fraction of agents of $Q$. To prove such an allocation exists, we present a randomized process. To this end, we define the concept of guiding graph in the following way:

\begin{definition}
	For a subset $Q$ of the agents and $k = \lfloor n/|Q| \rfloor$, the guiding graph is a bipartite graph with two parts random-seed nodes and allocation nodes that meets the following properties:
	\begin{itemize}
		\item The degree of each allocation node is exactly $k+1$. Also, each allocation node corresponds to one agent of $Q$ (the correspondence is many to one).
		\item Every random-seed node is incident to exactly $|Q|$ allocation nodes each corresponding to one agent of $|Q|$. More precisely, for every random-seed vertex $v$ and every agent $\agent_{x_i}$ in $Q$, $v$ has an edge to an allocation node that corresponds to agent $\agent_{x_i}$.
	\end{itemize}
\end{definition}

For reasons that we explained in Section~\ref{sec:o2}, our proof is based on a guiding graph that has a large girth. Thus, we first prove that for any integer $g$, there exists a guiding graph that has a girth of at least $g$. We remind the reader that such a guiding graph would have exponentially many vertices in both parts but the size of this graph plays no role in our analysis.
\begin{lemmarep}\label{lem:girth}
	For any integer $g$, there exists a guiding graph that has a girth of at least $g$.
\end{lemmarep}
\begin{appendixproof}
	We first point out that  a complete bipartite graph $\mathcal{G}$ with $k+1$ random-seed nodes and $|Q|$ allocation nodes each corresponding to one agent of $Q$ meets the conditions of our guiding graph.
	In what follows, we show that from any guiding graph $\mathcal{G}$, one can construct another guiding graph $\mathcal{G}'$ that has a larger girth. This proves that for any girth $g$, there exists a guiding graph whose girth is at least $g$.
	
	To this end, arbitrarily label the edges of $\mathcal{G}$ with numbers $2^0, 2^1, \ldots, 2^{r-1}$ where $r$ is the number of edges in $\mathcal{G}$. Also, let $h = 2^r$. We construct $\mathcal{G}'$ in the following way:
	\begin{itemize}
		\item For every vertex $v$ of $\mathcal{G}$, there are $h$ copies of $v$ in $\mathcal{G}'$ denoted by $v^0, v^1, v^2, \ldots, v^{h-1}$. Copies of the random-seed nodes of $\mathcal{G}$ are random-seed nodes of $\mathcal{G}'$ and copies of the allocation nodes of $\mathcal{G}$ are also allocation nodes of $\mathcal{G}'$. Moreover, for each allocation node of $\mathcal{G}$ which corresponds to an agent $\agent_{x_i} \in Q$, all copies of that node also correspond to agent $\agent_{x_i}$. 
		\item For each edge $(u,v)$ between an allocation node $u$ of $\mathcal{G}$ and a random-seed node $v$ of $\mathcal{G}$ whose label is $x$, we put an edge between $v^i$ and $u^{(i+x) \textsf{ mod } h}$ for all $0 \leq i \leq h$.
	\end{itemize}
	It follows from the construction of $\mathcal{G}'$ that it is a guiding graph. More precisely, each of its random-seed nodes are connected to $|Q|$ allocation nodes that have a 1-to-1 correspondence to the agents of $Q$ and also each allocation node of $\mathcal{G}'$ is connected to exactly $k+1$ random-seed vertices of $\mathcal{G'}$. We prove in the following that the girth of $\mathcal{G}'$ is greater than the girth of $\mathcal{G}$. To this end, let $v_{1}^{s_1}, u_{1}^{t_1}, v_{2}^{s_2}, u_{2}^{t_2}, \ldots, v_{c}^{s_c}, u_{c}^{t_c}$ be the shortest cycle of $\mathcal{G}'$ (odd nodes are random-seed nodes and even nodes are allocation nodes). This implies that $v_1, u_1, v_2, u_2,\ldots, v_c, u_c$ is also a closed walk of $\mathcal{G}$ with no edge traversed back and forth consecutively. We prove that this closed walk cannot be a cycle of $\mathcal{G}$ and therefore $\mathcal{G}$ has a smaller cycle which completes the proof. Thus, assume for the sake of contradiction that $v_1, u_1, v_2, u_2, \ldots, v_c, u_c$ is also a cycle in $\mathcal{G}$. This means that all the edge labels in this cycle are unique powers of $2$. Thus, if we take a subset of these edges and sum up the values of their labels and subtract the values of the rest of the labels of this cycle from the sum, the final value would be a non-zero number in range $[-h+1, h-1]$. Notice that there is a 1-to-1 correspondence between the edge weights of the corresponding cycle of $\mathcal{G}$ and numbers in series $S = (t_1 - s_1 + h) \text{ mod } h, (t_1 - s_2 + h) \text{ mod } h, (t_2 - s_2 + h) \text{ mod } h, (t_2 - s_3 + h) \text{ mod } h, \ldots, (t_c - s_c + h) \text{ mod } h, (t_c - s_1 + h) \text{ mod } h$. This implies that $(t_1-s_1 + h) \text{ mod } h + (s_2-t_1 + h) \text{ mod } h + (t_2 - s_2 + h) \text{ mod } h + \ldots + (s_1-t_c + h)\text{ mod } h$ which is divisible by $h$ can be made by multiplying the edge labels of the corresponding cycle in $\mathcal{G}$ by either $1$ or $-1$ and summing up their values. This is a contradiction. 
\end{appendixproof}

Now, we fix a small $\epsilon > 0$ and denote by $\mathcal{G}$ a guiding graph whose girth is at least $g = (10 (k+10)(|M|+10)(|Q|+10)+10)/\epsilon$.  Our proof is based on a randomized process on the guiding graph. In this randomized process, we label the allocation nodes and the edges of $\mathcal{G}$ with bundles of items in the following way:

\begin{enumerate}
	\item First, for each allocation node $v$ corresponding to agent $\agent_{x_i} \in Q$, we label $v$ with one of the $n$ bundles of $\agent_{x_i}$ whose value to her is at least $1$. After this labelling, the probability that each item appears in the label of a fixed allocation node is $1/n$ and thus the probability that an item appears in the label of any of the $|Q|$ incident vertices of a random-seed node is at most $|Q|/n \leq 1/k$.
	\item After the labels of the allocation nodes are realized, for each allocation node, we label each of its edges by a subset of the items in its label. We maintain the following two properties in our labelling: (1) The labels of the edges of each random-seed vertex will be disjoint. (2) for each allocation node, with probability at least $(1-\epsilon)$ none of the items of its label is missing in the label of more than one of its edges, i.e., each item in the label of an allocation node is present in the labels of at least $k$ (out of $k+1$) of its edges with probability at least $1-\epsilon$.
\end{enumerate}

While the vertex labelling in (i) is simple and uniformly at random, in order to meet the two properties stated in (ii) we utilize a rather complex method that we explain in the following. At this point, each allocation node of $\mathcal{G}$ has a label according to the procedure stated in (i). To explain our method for edge labelling of (ii), fix an item $\ite_x$ and define $\mathcal{H}_x$ as an induced subgraph of $\mathcal{G}$ that contains all the random-seed nodes but only the allocation nodes whose labels contain $\ite_x$. We first prove that for each allocation node $v$, with  probability at least $1-\frac{\epsilon}{|M|}$ either $v$ is not present in $\mathcal{H}_x$ or the connected component of $\mathcal{H}_x$ that includes $v$ is a tree.

\begin{lemmarep}\label{lemma:prob}
	Fix an item $\ite_x \in \items$ and an allocation node $v$ of $\mathcal{G}$ before the labelling process. After labelling the allocation vertices of $G$ and constructing $\mathcal{H}_x$ based on the labels of $\mathcal{G}$, with probability at least $1-\frac{\epsilon}{|M|}$, either $v$ is not present in $\mathcal{H}_x$ or the connected component of $v$ in $\mathcal{H}_x$ will be a tree.
\end{lemmarep}
\begin{appendixproof}
	We actually prove a stronger statement here. We show that if $v$ is present in $\mathcal{H}_x$, then with probability at least $1-\epsilon/|M|$, the connected component of $\mathcal{H}_x$ that contains $v$ is a tree. Thus, we assume for simplicity here that $v$ is indeed present in $\mathcal{H}_x$. We consider a fixed random-seed node $u$ of $\mathcal{H}_x$ and analyze the number of nodes we see if we run a BFS on this vertex. We start with one node and iterate over the edges of this node. The end point of each edge is an allocation node which is connected to $k+1$ random-seed nodes (including this one). Thus, each time we encounter an edge, we add one more allocation node and $k$ more random-seed nodes to our BFS queue (assuming we have not visited any of them). However, keep in mind that for the newly added random-seed nodes, we have already visited one of their neighbors (and thus the expected number of their unvisited neighbors will be multiplied by $\frac{|Q|-1}{|Q|}$). Thus, if we formulate the expected number of nodes (excluding $u$ itself) we visit in such a procedure by $f$, we obtain $$f \leq (1 + k + kf\frac{|Q|-1}{|Q|}) [\textsf{the expected degree of a random-seed node in }\mathcal{H}_x].$$ 
	On the other hand, due to the labelling of the allocation nodes, the expected number of times an item appears in the labels of any of the allocation nodes incident to a random-seed node is bounded by $|Q|/n \leq 1/k$. Thus, we obtain $f \leq 1 + (1 + k +  kf\frac{|Q|-1}{|Q|})/k$ which implies $f \leq 1 + 1/k + 1 +  f\frac{|Q|-1}{|Q|}$ and thus $f \leq (2 + 1/k)|Q|$. Therefore $1+f  \leq 1 + (2 + 1/k)|Q|$ is a bound on the expected size of the connected component of a random-seed node in $\mathcal{H}_x$. Since each allocation node is connected to $k+1$ random-seed nodes, this implies that the expected size of the connected component of $\mathcal{H}_x$ that includes $v$ is bounded by $$1 + (1 + (2 + 1/k)|Q|) (k+1) = 2 + k + 2(k+1)|Q| + \frac{|Q|(k+1)}{k} \leq 10 (|Q|+10) (k+10) .$$ Given that the girth of $\mathcal{G}$ as well as the girth of $\mathcal{H}_x$ is at least  $g = \frac{10 (k+10)(|M|+10)(|Q|+10)+10}{\epsilon}$ the probability that the size of the connected component of node $v$ in $\mathcal{H}_x$ is at least $g$ is at most $\epsilon/|M|$ and thus with probability at least $1-\epsilon/|M|$ the connected component of $v$ in $\mathcal{H}_x$ is a tree. 
\end{appendixproof}

Next we present an algorithm that determines the labels of which edges of $\mathcal{G}$ contain item $\ite_x$. We show that with this algorithm, for each random-seed vertex of $\mathcal{G}$ the label of at most one of its incident edges contains $\ite_x$. Moreover, we prove that for any allocation node $v$ of $\mathcal{G}$ with probability at least $1-\epsilon/|M|$ either $\ite_x$ is not present in the label of $v$ or item $\ite_x$ appears in the labels of $k$ (out of $k+1$) of its incident edges.

Let $Y$ be a subset of the allocation nodes of $\mathcal{H}_x$ whose connected components are trees. The number of edges incident to the vertices of $Y$ is exactly $(k+1)|Y|$. Since none of these edges appears in any cycle, there are at least $(k+1)|Y|+1$ distinct endpoints for these edges. $|Y|$ of these endpoints are the vertices of $Y$ and thus at least $k|Y|+1$ of these endpoints are random-seed nodes (Recall that $\mathcal{G}$ is bipartite). This basically means that the number of neighbors of $Y$ in the random-seed part is at least $k$ times the size of $|Y|$. By the Hall Theorem~\cite{hall1987representatives}, this means that there exist $k$ disjoint matchings of $\mathcal{H}_x$ each of which covers the allocation nodes whose connected components are trees. On the other hand, each random-seed vertex of $\mathcal{H}_x$ is covered in at most one of these matchings. Since the connected component of each allocation node of $\mathcal{H}_x$ is a tree with probability at least $1-\epsilon/|M|$, then every node of $\mathcal{H}_x$ is covered by at least $k$ matchings with probability at least $1-\epsilon/|M|$ and therefore for every allocation node $v$ of $\mathcal{G}$, with probability at least $1-\epsilon/|M|$, either $\ite_x$ does not appear in the label of node $v$, or it appears in the labels of $k$ of its incident edges.

We use the same algorithm to determine whether each item appears in the label of each edge of $\mathcal{G}$. For an allocation node $v$, the probability that each item of its label does not appear in the label of $k$ of its incident edges is bounded by $\epsilon/|M|$ and thus with probability at least $1-\epsilon$, each of the items in its label is present in the labels of at least $k$ of its incident edges. We  now prove Lemma~\ref{lemma:fourth}.

\begin{proof}[ of Lemma~\ref{lemma:fourth}]
	We explained the construction and labelling of $\mathcal{G}$ earlier. Here we prove that if we randomly select a random-seed node of $\mathcal{G}$ and construct the allocation based on the labels of the edges incident to that vertex, each agent $\agent_{x_i}$ of $Q$ has a $\frac{(1-\epsilon)k}{k+1}$ chance that she receives a bundle whose value to her is at least $1/2$. Let us be more precise about our allocation method. We randomly choose a random-seed vertex $v$ of $\mathcal{G}$. This vertex is incident to $|Q|$ allocation nodes of $\mathcal{G}$ that have a 1-to-1 correspondence with the agents of $Q$. Fix an agent $\agent_{x_i} \in Q$ and look at the edge incident to $v$ whose other endpoint is an allocation node corresponding to $\agent_{x_i}$. The bundle allocated to agent $\agent_{x_i}$ in our allocation is the label of that edge. Due to properties we explained earlier, our allocation is valid in the sense that no item is allocated to more than one agent. We show in the following that in expectation, at least a $(1-\epsilon) \frac{k}{k+1}$ fraction of the agents of $Q$ will receive a bundle whose value to them is at least $1/2$. Due to the fact that $\epsilon$ can be made arbitrarily small, this proves that there exists an allocation to agents of $Q$ in which at least a $\frac{k}{k+1}$ fraction of the agents of $Q$ receive a bundle whose value to them is at least $1/2$.
	
	Our proof is based on a double-counting argument. For each allocation node $u$ of $\mathcal{G}$  corresponding to agent $\agent_{x_i} \in Q$, denote by $L_u$ the label of node $u$. Based on our algorithm we have $\valu_{x_i}(L_u) \geq 1$. Color each edge $e$ incident to $u$ in red, if $\valu_{x_i}(L_e) \geq 1/2$ where $L_e$ is the label of edge $e$. We aim to prove that at least a $(1-\epsilon)\frac{k}{k+1}$ fraction of the edges of $\mathcal{G}$ are colored in red in expectation. To this end, fix an allocation vertex $u$ whose label is $L_u$ and denote by $e_1, e_2, \ldots, e_{k+1}$ its incident edges. Also, denote by $L_{e_j}$ the label of edge $e_j$. Define $W_1, W_2, \ldots, W_{k+1}$ as $k+1$ subsets of $L_e$, where $W_j = L_u \setminus L_{e_j}$. Our labelling method guarantees that with probability at least $1-\epsilon$  all $W_j$'s are disjoint. Let us assume that this is the case and assume for the sake of contradiction that $\valu_{x_i}(L_{e_j}) < 1/2$ holds for two distinct $j$'s. Let us without loss of generality assume that this is the case for $j=1$ and $j=2$ and therefore both $\valu_{x_i}(L_{e_1}) < 1/2$ and $\valu_{x_i}(L_{e_2}) < 1/2$ hold. Due to subadditivity of $\valu_{x_i}$, we have $\valu_{x_i}(W_1) >  1/2$ and $\valu_{x_i}(W_2) > 1/2$ and this contradicts our assumption given that $W_1 \subseteq L_{e_2}$ and $W_2 \subseteq L_{e_1}$. Thus, with probability at least $1-\epsilon$, at least $k$ out of $k+1$ edges incident to $u$ will be colored in red. This implies that in expectation, at least a $(1-\epsilon)\frac{k}{k+1}$ fraction of the edges of $\mathcal{G}$ are colored in red and therefore the probability that an agent $\agent_{x_i} \in Q$ receives a bundle in our algorithm whose value to her is at least $1/2$ is at least $(1-\epsilon)\frac{k}{k+1}$. This completes the proof.
\end{proof}

We show in the rest of this section that Lemma~\ref{lemma:fourth} yields the existence of an $1/O((\log \log )^2 n)$-\MMS\ guarantee for subadditive agents. This is very similar to what we explained earlier in Section~\ref{sec:warmup}. We start by setting $N_1 = N$ and use Lemma~\ref{lemma:prob} to construct an allocation of items to a subset of $N_1$ in which every agent receives a bundle whose value to her is at least $1/2$. Due to guarantees of Lemma~\ref{lemma:prob}, the size of this subset is at least $|N_1|/2$. We then construct $N_2$ by excluding those agents from $N_1$. Now, notice that $n/|N_2| \geq 2$. Thus, by Lemma~\ref{lemma:prob}, we can find an allocation of items to a subset of agents of $N_2$ in which each agent receives a bundle whose worth to her is at least $1/2$ and that the number of agents receiving  those bundles is at least $2/3 |N_2|$. We then continue this process by setting $N_3$ as $N_2$ minus those agents. This time we have $n/|N_3| \geq 6$ and thus Lemma~\ref{lemma:fourth} would provide a better allocation for us. Define $r_i$ as $n/|N_i|$ in this process. Thus we have $r_1 = 1$, $r_2 \geq  2$, $r_3 \geq 6$, and in general $r_{i+1} \geq r_i(r_i+1)$. It follows from this formula that $r_i \geq 2^{2^{i-2}}$ holds for $i \geq 2$ and thus our algorithm terminates after $\log \log n + 2$ steps. Since the items allocated in each step are disjoint and every agent receives a bundle which is worth at least $1/2$ to her in one of these steps, it follows that there exists a $(\log \log n + 2)$-multiallocation which is $1/2$-\MMS.

\begin{theoremrep}\label{theorem:main}
	The maximin share problem with subadditive agents admits a $\frac{1}{432000 (\log \log n)^2} \text{-}\MMS$ guarantee. 
\end{theoremrep}
\begin{appendixproof}
	For \( n = 2 \), a simple cut-and-choose algorithm guarantees \( 1 \)-\MMS\ to both agents. for \( n \geq 3 \), we compute the approximation guarantee of the algorithm.The algorithm terminates after at most \( \log \log n + 2 \) steps, which we simplify to at most \( 4 \log \log n \) steps for \( n \geq 3 \). This produces a \( 4 \log \log n \)-multiallocation where each agent values their bundle at least \( \frac{1}{2} \).
	Setting \( \alpha = 4 \log \log n \) and \( \eta = 2 \) in \Cref{approxl}, the approximation ratio is:
	\[
	\frac{1}{10800 \cdot 4 \log \log n \cdot 2 (\log(4 \log \log n) + \log \log n)} \text{-} \MMS.
	\]
	Now let us simplify the logarithmic term:
	$
	\log(4 \log \log n) + \log \log n \leq 4 \log \log n + \log \log n = 5 \log \log n.
	$
	Thus, the denominator becomes:
	$
	10800 \cdot 4 \cdot 2 \cdot 5 (\log \log n)^2 = 432000 (\log \log n)^2.
	$ Therefore, the algorithm guarantees a $$ \frac{1}{432000 (\log \log n)^2} \text{-}\MMS$$ allocation.
	
\end{appendixproof}

\printbibliography

@article{feige2009maximizing,
	title={On maximizing welfare when utility functions are subadditive},
	author={Feige, Uriel},
	journal={SIAM Journal on Computing},
	volume={39},
	number={1},
	pages={122--142},
	year={2009},
	publisher={SIAM}
}

@inproceedings{ghodsi2018fair,
	title={Fair Allocation of Indivisible Goods: Improvements and Generalizations},
	author={Ghodsi, Mohammad and HajiAghayi, MohammadTaghi and Seddighin, Masoud and Seddighin, Saeed and Yami, Hadi},
	booktitle={Proceedings of the 2018 ACM Conference on Economics and Computation},
	pages={539--556},
	year={2018},
	organization={ACM}
}

@inproceedings{garg2020improved,
	title={An improved approximation algorithm for maximin shares},
	author={Garg, Jugal and Taki, Setareh},
	booktitle={Proceedings of the 21st ACM Conference on Economics and Computation},
	pages={379--380},
	year={2020}
}

@inproceedings{caragiannis2016unreasonable,
	title={The unreasonable fairness of maximum Nash welfare},
	author={Caragiannis, Ioannis and Kurokawa, David and Moulin, Herv{\'e} and Procaccia, Ariel D and Shah, Nisarg and Wang, Junxing},
	booktitle={Proceedings of the 2016 ACM Conference on Economics and Computation},
	pages={305--322},
	year={2016},
	organization={ACM}
}

@article{Procaccia:first,
	title={Fair Enough: Guaranteeing Approximate Maximin Shares},
	author={Kurokawa, David and Procaccia, Ariel D and Wang, Junxing},
	journal={Journal of the ACM (JACM)},
	volume={65},
	number={2},
	pages={8},
	year={2018},
	publisher={ACM}
}

@inproceedings{seddighin2019externalities,
	title={Externalities and fairness},
	author={Seddighin, Masoud and Saleh, Hamed and Ghodsi, Mohammad},
	booktitle={The World Wide Web Conference},
	pages={538--548},
	year={2019}
}

@article{kaviani2024almost,
	title={Almost Envy-free Allocation of Indivisible Goods: A Tale of Two Valuations},
	author={Kaviani, Alireza and Seddighin, Masoud and Shahrezaei, AmirMohammad},
	journal={arXiv preprint arXiv:2407.05139},
	year={2024}
}

@inproceedings{farhadi2017fair,
	title={Fair Allocation of Indivisible Goods to Asymmetric Agents},
	author={Farhadi, Alireza and Hajiaghayi, MohammadTaghi and Ghodsi, Mohammad and Lahaie, Sebastien and Pennock, David and Seddighin, Masoud and Seddighin, Saeed and Yami, Hadi},
	booktitle={Proceedings of the 16th Conference on Autonomous Agents and MultiAgent Systems},
	pages={1535--1537},
	year={2017},
	organization={International Foundation for Autonomous Agents and Multiagent Systems}
}

@inproceedings{Saberi:first,
	title={On approximately fair allocations of indivisible goods},
	author={Lipton, Richard J and Markakis, Evangelos and Mossel, Elchanan and Saberi, Amin},
	booktitle={Proceedings of the 5th ACM conference on Electronic commerce},
	pages={125--131},
	year={2004},
	organization={ACM}
}

@article{Dubins:first,
	title={How to cut a cake fairly},
	author={Dubins, Lester E and Spanier, Edwin H},
	journal={American mathematical monthly},
	pages={1--17},
	year={1961},
	publisher={JSTOR}
}

@article{Steinhaus:first,
	title={The problem of fair division},
	author={Steinhaus, Hugo},
	journal={Econometrica},
	volume={16},
	number={1},
	year={1948}
}

@article{Budish:first,
	title={The combinatorial assignment problem: Approximate competitive equilibrium from equal incomes},
	author={Budish, Eric},
	journal={Journal of Political Economy},
	volume={119},
	number={6},
	pages={1061--1103},
	year={2011},
	publisher={JSTOR}
}

@article{amanatidis2015approximation,
	title={Approximation algorithms for computing maximin share allocations},
	author={Amanatidis, Georgios and Markakis, Evangelos and Nikzad, Afshin and Saberi, Amin},
	journal={ACM Transactions on Algorithms (TALG)},
	volume={13},
	number={4},
	pages={52},
	year={2017},
	publisher={ACM}
}

@book{brams1996fair,
	title={Fair Division: From cake-cutting to dispute resolution},
	author={Brams, Steven J and Taylor, Alan D},
	year={1996},
	publisher={Cambridge University Press}
}

@book{cormen2022introduction,
	title={Introduction to algorithms},
	author={Cormen, Thomas H and Leiserson, Charles E and Rivest, Ronald L and Stein, Clifford},
	year={2022},
	publisher={MIT press}
}

@article{stromquist1980cut,
	title={How to cut a cake fairly},
	author={Stromquist, Walter},
	journal={The American Mathematical Monthly},
	volume={87},
	number={8},
	pages={640--644},
	year={1980},
	publisher={Taylor \& Francis}
}

@article{even1984note,
	title={A note on cake cutting},
	author={Even, Shimon and Paz, Azaria},
	journal={Discrete Applied Mathematics},
	volume={7},
	number={3},
	pages={285--296},
	year={1984},
	publisher={Elsevier}
}

@article{maseed123,
	title={Improved maximin guarantees for subadditive and fractionally subadditive fair allocation problem},
	author={Seddighin, Masoud and Seddighin, Saeed},
	journal={Artificial Intelligence},
	volume={327},
	pages={104049},
	year={2024},
	publisher={Elsevier}
}

@inproceedings{dobzinski2024constant,
	title={A constant-factor approximation for nash social welfare with subadditive valuations},
	author={Dobzinski, Shahar and Li, Wenzheng and Rubinstein, Aviad and Vondr{\'a}k, Jan},
	booktitle={Proceedings of the 56th Annual ACM Symposium on Theory of Computing},
	pages={467--478},
	year={2024}
}

@article{hall1987representatives,
	title={On representatives of subsets},
	author={Hall, Philip},
	journal={Classic Papers in Combinatorics},
	pages={58--62},
	year={1987},
	publisher={Springer}
}

@book{robertson1998cake,
	title={Cake-cutting algorithms: Be fair if you can},
	author={Robertson, Jack and Webb, William},
	year={1998},
	publisher={AK Peters/CRC Press}
}

@inproceedings{akrami2024breaking,
	title={Breaking the 3/4 barrier for approximate maximin share},
	author={Akrami, Hannaneh and Garg, Jugal},
	booktitle={Proceedings of the 2024 Annual ACM-SIAM Symposium on Discrete Algorithms},
	pages={74--91},
	year={2024},
	organization={SIAM}
}

@inproceedings{akrami2023simplification,
	title={Simplification and improvement of MMS approximation},
	author={Akrami, Hannaneh and Garg, Jugal and Sharma, Eklavya and Taki, Setareh},
	booktitle={Proceedings of the Thirty-Second International Joint Conference on Artificial Intelligence},
	pages={2485--2493},
	year={2023}
}

@inproceedings{conitzer2017fair,
	title={Fair public decision making},
	author={Conitzer, Vincent and Freeman, Rupert and Shah, Nisarg},
	booktitle={Proceedings of the 2017 ACM Conference on Economics and Computation},
	pages={629--646},
	year={2017}
}

@article{akrami2024randomized,
	title={Randomized and deterministic maximin-share approximations for fractionally subadditive valuations},
	author={Akrami, Hannaneh and Mehlhorn, Kurt and Seddighin, Masoud and Shahkarami, Golnoosh},
	journal={Advances in Neural Information Processing Systems},
	volume={36},
	year={2024}
}

\end{document}